\documentclass[onefignum,onetabnum]{siamonline250106}

\usepackage{amsfonts}
\usepackage{graphicx}
\usepackage{epstopdf}
\ifpdf
  \DeclareGraphicsExtensions{.eps,.pdf,.png,.jpg}
\else
  \DeclareGraphicsExtensions{.eps}
\fi

\usepackage{enumitem}
\setlist[enumerate]{leftmargin=.5in}
\setlist[itemize]{leftmargin=.5in}

\usepackage[sort,compress,square,numbers]{natbib}

\newsiamremark{remark}{Remark}
\newsiamremark{defi}{Definition}
\newsiamremark{example}{Example}
\newsiamremark{hypothesis}{Hypothesis}
\crefname{hypothesis}{Hypothesis}{Hypotheses}
\newsiamthm{claim}{Claim}

\usepackage{lmodern}

\usepackage{microtype} 

\usepackage{url} 

\usepackage[version=3]{mhchem} 

\usepackage{amsmath}
\usepackage{amssymb}

\usepackage{graphicx}
\usepackage{caption}
\usepackage{subcaption}
\usepackage{float}
\usepackage{wrapfig}

\usepackage{mathrsfs} 

\usepackage{bm}

\usepackage{bbm}

\usepackage[perpage]{footmisc}

\usepackage{tikz-cd}

\usepackage[all]{xy}

\usepackage{bookmark}
\hypersetup{bookmarksnumbered}

\usepackage[nameinlink]{cleveref}
\crefname{ex}{Example}{Examples}
\crefname{thm}{Theorem}{Theorems} 
\crefname{lem}{Lemma}{Lemmas}
\crefname{prop}{Proposition}{Propositions}
\crefname{cor}{Corollary}{Corollaries} 
\crefname{con}{Conjecture}{Conjectures} 
\crefname{def}{Definition}{Definitions}
\crefname{rmk}{Remark}{Remarks}
\crefname{thmalph}{Theorem}{Theorems}

\newtheorem{theoremalphabetic}{Theorem}

\DeclareMathOperator{\rk}{rk} 
\DeclareMathOperator{\im}{im} 
\DeclareMathOperator{\diag}{diag} 

\newcommand{\D}{\mathcal{D}}

\newcommand{\X}{\mathcal{X}}

\newcommand{\A}{\mathcal{A}}
\newcommand{\Z}{\mathcal{Z}}
\newcommand{\I}{\mathcal{I}}
\newcommand{\U}{\mathcal{U}}
\renewcommand{\P}{\mathcal{P}}

\newcommand{\ZZ}{\ensuremath{\mathbb{Z}}}
\newcommand{\QQ}{\ensuremath{\mathbb{Q}}}
\newcommand{\RR}{\ensuremath{\mathbb{R}}}
\newcommand{\CC}{\ensuremath{\mathbb{C}}}

\newcommand{\VV}{\mathbb{V}}

\DeclareTextFontCommand{\bfemph}{\bfseries\em}
\newcommand{\term}{\bfemph}

\renewcommand{\k}{\kappa}

\headers{The generic geometry of steady state varieties}{Elisenda Feliu, Oskar Henriksson, and  Beatriz Pascual-Escudero}

\title{The generic geometry of steady state varieties}

\author{Elisenda Feliu\thanks{Department of Mathematical Sciences, University of Copenhagen
  (\email{efeliu@math.ku.dk}).}
\and Oskar Henriksson\thanks{Department of Mathematical Sciences, University of Copenhagen. Current address: Max Planck Institute of Molecular Cell Biology and Genetics, Dresden.
  (\email{oskar.henriksson@mpi-cbg.de}).}
\and Beatriz Pascual-Escudero\thanks{Departamento de Matemática e Informática aplicadas a las Ingenierías Civil y Naval
ETSI Caminos, Canales y Puertos, Universidad Politécnica de Madrid
  (\email{beatriz.pascual@upm.es}).}}

\usepackage{amsopn}

\ifpdf
\hypersetup{
  pdftitle={The generic geometry of steady state varieties},
  pdfauthor={Elisenda Feliu, Oskar Henriksson, and  Beatriz Pascual-Escudero}
}
\fi

\begin{document}

\maketitle

\begin{abstract}
We answer several fundamental geometric questions about reaction networks with power-law kinetics, on topics such as generic finiteness of the number of steady states,  robustness, and non\-degenerate multistationarity. In particular, we give an ideal-theoretic characterization of generic absolute concentration robustness, as well as conditions under which a network that admits multiple steady states also has the capacity for nondegenerate multistationarity. The key tools underlying our results come from the theory of vertically parametrized systems, and include a linear algebra condition that characterizes when the steady state system has positive nondegenerate zeros.
\end{abstract}



\section{Introduction} 
A fundamental object of interest in the study of reaction networks is the set of \emph{positive~steady~states}, which under the assumption of mass-action kinetics (or more generally power-law kinetics) are the positive zeros of a  \emph{vertically parametrized system}
\[
   f_\k(x):=N(\k\circ x^M)\,.
\]
Here, the variables $x=(x_1,\ldots,x_n)$ are the concentration of the species participating in the network, the parameters $\k=(\k_1,\ldots,\k_m)$ are rate constants for the reactions of the network, the matrix $N\in\RR^{n\times m}$ is the stoichiometric matrix, and $x^M$ is a vector of $m$ monomials, encoded by the columns of the kinetic matrix $M\in\ZZ^{n\times m}$.

The term \emph{vertically parametrized} was coined in \cite{helminck2022generic} and refers to the fact that each parameter $\kappa_i$ can appear several times in the system, but always in front of the same monomial. Thus, these systems generalize the more well-studied class of sparse systems with freely varying coefficients. A general framework for studying the generic properties of vertically parametrized systems was developed in \cite{FeliuHenrikssonPascual2023} (see also the subsequent papers \cite{HelminckHenrikssonRen2024,FeliuHenrikssonToric}), and forms the foundation for the present paper.

The ordinary differential equations $\frac{dx}{dt} =f_\k(x) $  that model a mass-action system often have linear first integrals that are independent of the choice of reaction rate constants. These correspond to conservation laws of the form $Lx=b$ where the rows of $L\in\RR^{d\times n}$ form a basis for the left kernel of $N$ and $b=(b_1,\ldots,b_d)$ are total amounts. Thus, the positive steady states compatible with a choice of total amounts $b$ are the zeros of the polynomial system
\begin{equation}
\label{eq:augmented_system_intro}
    F_{\k,b}(x):=\left(\begin{array}{l}N(\k\circ x^M)\\Lx-b\end{array}\right).
\end{equation}

Many questions about the steady states of a reaction network are related to understanding  how the geometry of the positive zero sets $\VV_{>0}(f_\k)$ and $\VV_{>0}(F_{\k,b})$ in $\RR^n_{>0}$ depends on the parameters $\k\in\RR^m_{>0}$ and $b\in\RR^d$.
Even though the study of reaction networks in the current mathematical formalism goes back at least to Feinberg, Horn and Jackson in the 1970's \cite{Feinberg1972complex,hornjackson}, where the theory was developed with a main focus on reaction networks arising in chemistry, many fundamental properties about the underlying polynomial systems are still not fully understood, including properties such as dimension, finiteness, and singularities. 
These concepts play an important role in applying machinery from algebraic geometry 
(see, e.g., \cite[Section~6.5]{P-M}, \cite{pascualescudero2020local}), as well as in lifting properties of small networks to larger networks containing the small ones  as submotifs (see, e.g., \cite{pantea-banaji,craciun-feinberg,cappelletti:flows,feliu:intermediates,joshi-shiu-II}).

Understanding the geometry of $\VV_{>0}(f_\k)$ and $\VV_{>0}(F_{\k,b})$ for all possible networks and  parameter values is a very challenging task. Indeed, it follows from the classical \emph{Hungarian lemma} \cite{Hars1979inverse} that the positive part of \emph{any} algebraic variety can appear as $\VV_{>0}(f_\k)$ for some network and some choice of $\k\in\RR^m_{>0}$, which means that the set of steady states can display intricate geometry, such as its dimension being larger or even lower than expected. However, one of the main messages of this paper is that the problem becomes much more well-behaved if we change it to instead understand  $\VV_{>0}(f_\k)$ and $\VV_{>0}(F_{\k,b})$ \emph{up to perturbations} of the parameter values (i.e., in \emph{open regions} of parameter space). In particular, we show that for many properties, the behavior up to perturbation agrees qualitatively with the behavior of the complex varieties $\VV_{\CC^*}(f_\k)$ and $\VV_{\CC^*}(F_{\k,b})$ in $(\CC^*)^n$ for generic $(\k,b)\in\CC^{m+d}$, and hence lends itself to be studied with tools from algebraic geometry.

We first tackle the question of \emph{finiteness}. By a simple equation count, it is reasonable to expect that $\VV_{>0}(f_\k)$ has codimension $\rk(N)$, and that $\VV_{>0}(F_{\k,b})$ should be finite. Nevertheless, it is easy to construct examples of networks where these sets have higher-than-expected dimension for some parameter values. In fact, this can be done even for very well-behaved families of networks, such as those that are endotactic \cite{KothariDeshpande24} or weakly reversible 
\cite{boros2020weakly} (see \Cref{ex:BCY}). However, it has been an open question whether this type of pathology can arise in regions of positive measure in parameter space (see \cite[Section~5]{boros2020weakly}).

As our first main result, we answer this question in the negative. More precisely, we show that there are two possible scenarios for $\VV_{>0}(F_{\k,b})$, depending on whether or not the network is \emph{nondegenerate}, in the sense that there exist parameter values $(\k,b)$ and a positive zero $x^*$ of $F_{\k,b}$  such that the Jacobian  of $F_{\k,b}$ at $x^*$ is nonsingular. This is equivalent to requiring that 
the matrices $N$, $M$ and $L$ from \eqref{eq:augmented_system_intro} satisfy $\ker(N)\cap\RR^m_{>0}\neq\varnothing$ and
\begin{equation}\label{eq:intro_rank_condition}
\rk \begin{bmatrix}N \diag(w)M^\top \diag(h) \\ L\end{bmatrix} =n\quad\text{for some $(w,h)\in \ker(N) \times (\CC^*)^n$}.
\end{equation}

\begin{theoremalphabetic}[{\Cref{thm:Fnet}}]
\label{thmA:nondegeneracy_result}
For a reaction network with $\ker(N)\cap\RR^m_{>0}\neq\varnothing$, let  $\Z_\mathrm{cc}$ be the  solvability locus, consisting of   $(\k,b)\in\RR^m_{>0}\times\RR^d$ for which $\VV_{>0}(F_{\k,b})$ is nonempty. If the network is \textbf{nondegenerate}, then:
\begin{itemize}
    \item $\Z_\mathrm{cc}$ has nonempty interior.
    \item $\VV_{>0}(F_{\k,b})$ is finite for generic $(\k,b)\in\Z_\mathrm{cc}$.
    \item The Jacobian of $F_{\k,b}$ is nonsingular on $\VV_{>0}(F_{\k,b})$ for generic $(\k,b)\in\Z_\mathrm{cc}$.  
\end{itemize}
\pagebreak
If the network is \textbf{degenerate}, then
\begin{itemize}
    \item $\Z_\mathrm{cc}$ is nonempty but has empty interior.
    \item The Jacobian of $F_{\k,b}$ is singular on  all of $\VV_{>0}(F_{\k,b})$ for all $(\k,b)\in\Z_\mathrm{cc}$.
\end{itemize}
\end{theoremalphabetic}

In particular, \Cref{thmA:nondegeneracy_result} tells us that a reaction network cannot have infinitely many steady states in a nonempty open region in the space of rate constants and total amounts.

Analogous statements relating emptiness and dimension are 
given for $\VV_{>0}(f_\k)$ in \Cref{thm:fnet}. Both of these results are an application of the theory of vertically parametrized systems developed in \cite{FeliuHenrikssonPascual2023}. Condition \eqref{eq:intro_rank_condition} is easy to check computationally for specific networks, which we demonstrate in \Cref{subsec:computational}, where we check it for all networks in the database ODEbase \cite{odebase}. In \Cref{subsec:classes}, we also prove that it is satisfied by large families of networks, including weakly reversible networks, injective networks, conservative networks lacking boundary steady states, and the networks satisfying the hypotheses of the deficiency one theorem.

Next, we treat the problem of \term{absolute concentration robustness} (ACR). The term ACR  was introduced in \cite{shinar-science} and has been extensively studied, e.g., \cite{ACR_Alon,ACR_Cap,ACR_AlgGeom,invariantsACR, ACRdim1,pascualescudero2020local}. We say that a network has ACR for $X_i$ for a given $\k\in\RR^m_{>0}$ if $\VV_{>0}(f_\k)$ is nonempty and contained in a parallel translate of the $i$-th coordinate  hyperplane.
The weaker notion of \term{local ACR} was introduced in \cite{pascualescudero2020local} and means that $\VV_{>0}(f_\k)$ is contained in a finite union of translates of a coordinate hyperplane.

Understanding when  ACR or local ACR arise for all $\k\in\RR^m_{>0}$ is a very challenging algebraic-geometric problem, as has recently been explored in detail in \cite{ACR_AlgGeom,pascualescudero2020local}, but the \emph{generic} counterpart of these problems, where we only require that a property holds for almost all $\k\in\RR^m_{>0}$, turns out to be much more well-behaved. Using the theory of vertically parametrized systems, we prove that the rank condition from \cite[Section~5]{pascualescudero2020local} precisely characterizes generic local ACR. In particular, it is a necessary condition for ACR to hold in a nonempty open subset of parameter space. We also strengthen the sufficient ideal-theoretic condition from \cite[Proposition~3.8]{ACR_AlgGeom} to a complete characterization of generic ACR.

\begin{theoremalphabetic}[{\Cref{thm:localACR}, \Cref{cor:ACRgen}}]
\label{thmA:acr_results}
For a nondegenerate network,  the following are equivalent:
\begin{enumerate}[label=(\roman*)]
    \item The network has generic local ACR for $X_i$. 
    \item $\rk( (N\diag(w)M^\top)_{\setminus i}) < \rk(N)$ for all $w \in \ker(N)$,  where {\footnotesize ${\setminus i}$} indicates removal of the $i$-th column.
    \item There exists a nonconstant polynomial 
 \[ g \in \left(\big\langle N(\k\circ x^M) \big\rangle: (x_1 \cdots x_n)^\infty \right) \cap \RR(\k_1,\ldots,\k_m)[x_i]\, . \]
\end{enumerate}
\vspace{-0.7cm}
Furthermore:

\begin{itemize}
    \item The network has generic ACR for $X_i$ if $g$ has a single positive root for generic $\k\in\RR^m_{>0}$.
    \item If the network does not have generic local ACR for $X_i$, then, for generic $\k\in\RR^m_{>0}$, the network does not have (local) ACR for $X_i$.
\end{itemize}
\end{theoremalphabetic}

Finally, we also study the property of \term{multistationarity}, which refers to 
$\VV_{>0}(F_{\k,b})$ having at least two elements for some $(\k,b)\in\RR^m_{>0}\times\RR^d$. When, in addition, the steady states are stable, this property is believed to underlie cellular decision processes, in that it offers robust switch-like behavior as a response to changes in the parameters via a phenomenon known as \emph{hysteresis} \cite{laurent1999}.  The study  of multistationarity can be found at the roots of reaction network theory, with celebrated results on complex balancing \cite{Feinberg1972complex,FEINBERG19872229,hornjackson} and numerous algorithms and criteria to decide upon its existence or lack thereof, e.g., \cite{craciun-feinbergI,conradi-PNAS,CRNToolbox,PerezMillan}; see \cite{joshi-shiu-III} for an overview.

Several results for inferring multistationarity from reduced models rely on the implicit function theorem or homotopy continuation \cite{pantea-banaji,craciun-feinberg,cappelletti:flows,feliu:intermediates,joshi-shiu-II,tang:one-dim}, and, therefore,
require a stronger version of multistationarity, where there is a choice of parameters for which $\VV_{>0}(F_{\k,b})$ has at least two positive steady states that are \emph{nondegenerate}, in the sense that the Jacobian of $F_{\k,b}$ is nonsingular.  This property is referred to as \term{nondegenerate multistationarity}. 

In \cite[Conjecture 2.3]{Joshi2017small}, the authors conjecture that if $\VV_{>0}(F_{\k,b})$ is finite for all $(\k,b)\in\RR^m_{>0}\times\RR^d$, and has cardinality $\ell$ for some choice of parameters, then there is also a choice of parameters such that the network has $\ell$ \emph{nondegenerate} positive steady states. This is known as the \emph{Nondegeneracy Conjecture}. It has been proven for small networks (with at most $2$ species and $2$ reactions, which can be reversible) in \cite{Joshi2017small,shiuwolff:small} and for $\rk(N)=1$ in \cite{tang:one-dim}, but the general case remains open.
 Here, we prove the conjecture when $\ell=2$ 
under some additional assumptions, which include the case where the network is nondegenerate.  
A weaker version of our result, but easier to state, is as follows (see \Cref{thm:nondegeneracy_conjecture} for the full statement, which allows for milder assumptions on the network and a more precise lower bound of the number of positive steady states). 

\begin{theoremalphabetic}[\Cref{thm:nondegeneracy_conjecture}, \Cref{cor:nondegeneracy_conjecture}]
\label{thmA:nondegeneracy_conjecture}
   If a nondegenerate network admits multistationarity and $\VV_{>0}(F_{\k,b})$ is finite for all $(\k,b)\in\RR^m_{>0}\times\RR^d$, then the network also admits   nondegenerate multistationarity.
\end{theoremalphabetic}

The paper is organized as follows. 
In \Cref{sec:notations_defs}, we fix the notation for the rest of the paper, and recall some basic terminology of chemical reaction network theory, as well as some results regarding the Jacobian of $f_\k$ and $F_{\k,b}$. In \Cref{sec:theorems}, we discuss the connection between the properties of nondegeneracy, dimension, and nonemptiness, leading up in particular to   \Cref{thmA:nondegeneracy_result}. We also discuss some  computational aspects, as well as how nondegeneracy relates to other properties of reaction networks such as weak reversibility and the dimension of the kinetic subspace. In \Cref{sec:ACR}, we address ACR, state \Cref{thmA:acr_results} and give several examples to clarify the relation to  previous work. After this, we devote \Cref{sec:multi} to nondegenerate multistationarity and our proof of \Cref{thmA:nondegeneracy_conjecture}. Finally, we give the proofs relying on technical aspects of vertically parametrized systems in \Cref{sec:thm_proofs}.

\section{Reaction networks and steady states}\label{sec:notations_defs}

In this section, we give a quick overview of reaction networks, introduce the varieties of positive steady states, and fix notation that will be used in the rest of the paper.

A \term{reaction network} is simply a collection of \emph{reactions}
\[\alpha_{1j}X_1+\dots+\alpha_{nj} X_n \ce{->[\k_j]} \beta_{1j}X_1+\dots+\beta_{nj} X_n,\quad j=1,\dots,m\,, \]
that model interactions among \emph{species} $X_1,\ldots ,X_n$. Each of these interactions transforms a \emph{complex} (a $\ZZ_{\geq 0}$-linear combination of the species) called the \emph{reactant}, into another complex called the \emph{product}.  The reactions are labelled by positive real numbers $\k_j>0$ called \emph{reaction rate constants}, which play a role in the dynamics of the system represented by the network. 

This way, a reaction network can be considered as a digraph, with reactions as directed edges, labeled by the reaction rate constants, and complexes as nodes.

\begin{example}
\label{ex:calcium_network}
A simple representation of enzymatic transfer of calcium ions between the endoplasmic reticulum and the cytosol gives rise to the following network with $n=4$ species and $m=6$ reactions:
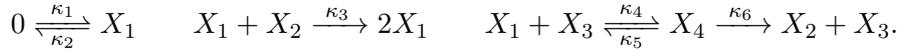
\[
\begin{aligned}
0 & \ce{<=>[\k_1][\k_2]}  X_1 \quad&
X_1+X_2 &  \ce{->[\k_3]} 2X_1\quad &
X_1+X_3 &  \ce{<=>[\k_4][\k_5]} X_4 \ce{->[\k_6]} X_2+X_3\,.
\end{aligned}
\]
Here, $X_1$ stands for cytosolic calcium, $X_2$ for calcium in the endoplasmic reticulum, and $X_3$ is an enzyme catalyzing the transfer via the formation of an intermediate $X_4$ \cite{ges05}. 
\end{example}

Under the assumption of \emph{power-law kinetics} with \term{kinetic matrix} $M=(\nu_{ij}) \in \RR^{n\times m}$,
the evolution of the concentration $x=(x_1,\dots,x_n)$ of the species $X_1,\dots,X_n$ over time is given by the autonomous system of ordinary differential equations (ODEs)
\begin{equation}\label{eq:ode}
\frac{dx}{dt} =N (\k \circ x^M)\,, \qquad x\in \RR^n_{>0}\,,
\end{equation}
where $N=(\beta_{ij}-\alpha_{ij})\in \ZZ^{n\times m}$ is the \term{stoichiometric matrix}, $x^M$ denotes the vector whose $j$-th entry is the product 
$x_1^{\nu_{1j}}\cdots x_n^{\nu_{nj}}$ (i.e., the monomial in $x$ with exponents given by the $j$-th column of $M$), and $\k \circ x^M$ is the entry-wise product of the two vectors $\k=(\k_1,\ldots ,\k_m)$  and $x^M$. 
For the specific case in which the kinetics is \term{mass action}  \cite{Guldberg-Waage},  $M=(\alpha_{ij})\in \ZZ_{\geq 0}^{n\times m}$ is the \emph{reactant matrix}. 
In the examples, we consider mass-action kinetics unless stated otherwise.

For the rest of our work, we will assume that $M\in \ZZ^{n\times m}$, i.e., that $M$ has integer entries, when referring to power-law kinetics. The results  extend to systems where $M\in \QQ^{n\times m}$ using the approach in \cite[Section 4.2]{pascualescudero2020local}. 
As we allow negative exponents, we mostly work in the Laurent polynomial ring $\RR[x^\pm]= \RR[x_1^\pm,\dots,x_n^\pm]$.

Observe that \eqref{eq:ode} is well defined also over $\RR^n_{\geq 0}$ if $M$ has nonnegative entries. Under mild additional assumptions, namely that 
all negative terms of the expression for $dx_i/dt$ are multiples of $x_i$,  both $\RR^n_{\geq 0}$ and $\RR^n_{> 0}$ are forward invariant \cite{volpert}.
In particular, this is always the case  under mass-action kinetics.

A (positive) \term{steady state} of the ODE system in \eqref{eq:ode} is a tuple $x=(x_1,\ldots ,x_n)\in \RR^n_{>0}$ such that $N (\k \circ x^M) =0$.
The steady states provide useful information about the dynamics of the biological system under study, and this is a key topic in the theory of reaction networks; see \cite{feinberg2019foundations} for an introduction to the field.

To remove redundancies arising when $N$ does not have full rank $n$, we make a choice of matrix $C\in \RR^{s\times m}$ with $s:=\rk(N)$ and $\ker(C)=\ker(N)$, and consider the \term{steady state system}
\begin{equation}\label{eq:f}
    f  := C (\kappa \circ {x}^M)  \in \RR[\k,x^\pm]^s.
\end{equation}
The results in this work do not depend on the specific matrix $C$, so a choice is implicitly made throughout. 
 We write $f_{\k}$ for the system \eqref{eq:f} evaluated at a fixed $\k \in \RR^m_{>0}$.  The \term{positive steady state variety} of the network is then the (semialgebraic) set
\[\VV_{>0}(f_\k):=\{x\in\RR^n_{>0} : C(\k\circ x^M)=0\} = \{x\in\RR^n_{>0} : N(\k\circ x^M)=0\}\,.\]

We define the \term{solvability locus} $\Z \subseteq \RR^m_{>0}$ of $f$ to be the set of parameter values for which the network admits positive steady states:

\[\Z := \{ \k\in \RR^m_{>0} : \VV_{>0}(f_\k) \neq \varnothing\}\,.\]

By letting $\RR^*$ and $\CC^*$  denote the set of real and complex numbers excluding $0$, we will also consider the  \emph{real} and \emph{complex steady state varieties}, given by
\[\VV_{\RR^*}(f_\k):=\{x\in(\RR^*)^n: f_\k(x) =0\} \qquad \text{and}\qquad \VV_{\CC^*}(f_\k):=\{x\in(\CC^*)^n: f_\k(x) =0\}. \]
Note that $\VV_{>0}(f_\k)= \VV_{\CC^*}(f_\k) \cap \RR^n_{>0}$.

The vector subspace $\im(N)$ is called the \term{stoichiometric subspace}. The trajectory of  system \eqref{eq:ode} with initial condition $x^0$ is confined to the linear subspace $x^0+\im(N)$. We will therefore also be interested in the positive steady states constrained to these linear subspaces. Specifically, we consider a fixed matrix $L\in \RR^{d\times n}$ with full rank $d:=n-s$ (recall $s=\rk(N)=\rk(C)$), whose rows form a basis of the left kernel of $N$. Then, we consider the (open) \term{stoichiometric compatibility classes}
\[ \mathcal{P}_{b}:=\{ x \in \RR^n_{> 0} : Lx-b=0\}, \qquad b\in \RR^{d}.\]
 Such a matrix $L$ is called a \term{matrix of conservation laws}. The choice of $L$ is implicit throughout, and it does not affect the conclusions of this work.

Many questions about the dynamics of \eqref{eq:ode}, and in particular about the steady states, are studied in the restriction to the stoichiometric compatibility classes. 
Two steady states in the same class are said to be  stoichiometrically compatible. 
Studying the set of all   positive steady states in the class $\mathcal{P}_b$ can be done via the square polynomial system
\begin{equation}\label{eq:F}
F  =\left( \begin{array}{l}C(\kappa \circ {x}^M) \\ Lx - b \end{array}\right)  \in \RR[\k,b,x^\pm]^n,
\end{equation} 
which we refer to as the \term{augmented steady state system} (by $L$). 
We write $F_{\k,b}$ for the system \eqref{eq:F} when $\k$ and $b$ have been fixed.

By construction, 
\[\VV_{>0}(F_{\k,b}) = \VV_{>0}(f_\k) \cap  \mathcal{P}_b \, . \]
Analogous to $\Z$, 
we consider   the \term{solvability locus} $\Z_{\rm cc} \subseteq \RR^m_{>0}\times \RR^d$ of $F$ to be 
\begin{align*}
    \Z_{\rm cc} :=&\  \{ (\k,b)\in \RR^m_{>0}\times \RR^d : \VV_{>0}(F_{\k,b})\neq \varnothing\}\, , 
\end{align*}
where the subscript cc is short for ``compatibility class''. 
Note here that  $\VV_{>0}(f_\k)\cap \P_{b}\neq \varnothing$ requires that $b\in L(\RR^n_{>0})$, that is, $b$ belongs to the positive cone  on the columns of $L$. 
We remark that as $L$ has full rank, $L(\RR^n_{>0})$ is a $d$-dimensional cone.

The sets $\VV_{\CC^*}(F_{\k,b})$ and $\VV_{\RR^*}(F_{\k,b})$ are defined analogously to $f$. 
We note that $F=f$ when $s=n$.

\begin{example}
\label{ex:calcium_network2}
For the network in \Cref{ex:calcium_network}, the ODE system \eqref{eq:ode} is
built out of the matrices 
\[N= \begin{bmatrix}
 1 & -1 & 1 & -1 & 1 & 0 \\
 0 & 0 & -1 & 0 & 0 & 1 \\
 0 & 0 & 0 & -1 & 1 & 1 \\
  0 & 0 & 0 & 1 & -1 & -1 
\end{bmatrix}\qquad \text{and}\qquad M= \begin{bmatrix}
  0 & 1 & 1 & 1 & 0 & 0 \\
 0 & 0 & 1 & 0 & 0 & 0 \\
 0 & 0 & 0 & 1 & 0 & 0 \\
  0 & 0 & 0 & 0 & 1 & 1 
 \end{bmatrix}.  \]
The rank of $N$ is $s=3$, as the bottom two rows are linearly dependent. Hence, we can choose
\[C= \begin{bmatrix}
 1 & -1 & 1 & -1 & 1 & 0 \\
 0 & 0 & -1 & 0 & 0 & 1 \\
 0 & 0 & 0 & -1 & 1 & 1
 \end{bmatrix} \qquad \text{and}\qquad L= \begin{bmatrix}
  0 & 0 &  1 & 1
 \end{bmatrix} \, .
\]
With this in place, the steady state system becomes
\begin{equation*}
\label{eq:F_calcium}
f= (\k_1-\k_2x_1+\k_3x_1x_2-\k_4x_1x_3+\k_5x_4, 
 -\k_3x_1x_2+\k_6x_4,  -\k_4x_1x_3+\k_5x_4+\k_6x_4), 
\end{equation*}
while the augmented steady state  system has the additional entry $ x_3+x_4-b$. 
\end{example}

A key tool in the study of parametrized polynomial systems is the \term{incidence variety}, which consists of all pairs of parameter and variable values that together solve the system. For the augmented steady state system $F$, it admits the bijective parametrization

\begin{equation}\label{eq:parametrization_of_incidence_variety}
\begin{aligned}
\phi \colon \ker(C)\times (\CC^*)^n  & \rightarrow   \{(\k,b,x) \in \CC^m \times \CC^d \times (\CC^*)^n : F(\k,b,x)=0 \}
\\ (w,h) & \mapsto (w\circ h^M,Lh^{-1},h^{-1})\,,
\end{aligned}
\end{equation}
where $h^{-1}$ is taken componentwise \cite[Theorem~3.1]{FeliuHenrikssonPascual2023}.
Importantly, $\phi$ also restricts to a bijection onto the positive incidence variety,  in which case $(w,h)$ are referred to \emph{convex coordinates} \cite{Clarke1980stability}:
\begin{equation}\label{eq:parametrization_of_incidence_variety_pos}(\ker(C)\cap \RR^m_{>0})\times \RR^n_{>0}\to\{(\k,b,x) \in \RR_{>0}^m \times \RR^d \times \RR_{>0}^n : F(\k,b,x)=0 \}.
\end{equation}

This bijection allows us to derive several foundational facts about the steady state systems, including nonemptiness and connectivity of solvability loci, and a well-known parametrization of the set of Jacobians for all parameter values and zeros. In preparation for the latter, we define the matrices

\begin{equation}
\begin{aligned}\label{eq:M_f_RN}
Q_f(w)&:=C \diag(w)M^\top \in \CC^{s\times n}, \quad w\in\CC^m,\\ Q_F(w,h)&:=\left[ \begin{array}{c}C \diag(w)M^\top \diag(h) \\ L \end{array}\right]\in \CC^{n\times n},\quad w\in\CC^m,\:\: h\in(\CC^*)^n\mbox{.} 
\end{aligned}
\end{equation}

\begin{proposition}
\label{prop:Jacmatrices}
Consider a reaction network with stoichiometric matrix $N \in \ZZ^{n\times m}$,  kinetic matrix $M\in \ZZ^{n\times m}$, steady state system  $f$ as in \eqref{eq:f} for $C\in \RR^{s\times n}$ of full rank $s=\rk(N)$ such that $\ker(N)=\ker(C)$, and augmented steady state system $F$  as in \eqref{eq:F} for a full rank matrix $L\in \RR^{d \times n}$ with $LN=0$ and $d=n-s$. Then the following holds:
\begin{enumerate}[label=(\roman*)]
\item $\ker(C)\cap \RR^m_{>0} \neq \varnothing \iff \Z \neq \varnothing \iff \Z_{\rm cc}\neq \varnothing$.   
\item The solvability loci $\Z_{\rm cc}$ and $\Z$ are connected when they are nonempty.
\item  For $ (\k,b,x)= \phi(w,h)$ it holds that
\[
 J_{f_{\k}}(x)=Q_f(w)\diag(h) \quad  \text{and}\quad J_{F_{\k,b}}(x)=Q_F(w,h). 
\]
\end{enumerate}
\end{proposition}
\begin{proof}
(i) For $\ker(C)\cap \RR^m_{>0}\neq \varnothing \Rightarrow \Z \neq \varnothing$ apply $\phi$ to some $(w,h) \in (\ker(C)\cap \RR^m_{>0}) \times \RR^n_{>0}$. The rest of the implications are immediate from definition. (ii) Follows from the connectivity of of $(\ker(C)\cap \RR^m_{>0})\times \RR^n_{>0}$, the continuity of $\phi$ in \eqref{eq:parametrization_of_incidence_variety_pos}, as well as the continuity of the the projections onto $\RR^m_{>0}\times \RR^d$ and $\RR^m_{>0}$.
(iii) Follows from a direct computation of the Jacobian matrices, see \cite[Proposition~3.2]{FeliuHenrikssonPascual2023}. 
\end{proof}

It is common in the literature to say that a network is \term{consistent} (or dynamically nontrivial) if the equivalent statements in \Cref{prop:Jacmatrices}(i) hold. 
 It follows from \Cref{prop:Jacmatrices}(iii) that we have the following equalities of sets:
\begin{align*}
\{J_{f_{\k}}(x):\k \in \RR_{>0}^m, x\in \VV_{>0}(f_\k)\}&=\{Q_f(w) \diag(h): (w,h)\in (\ker(C)\cap\RR^m_{>0})\times\RR_{>0}^n\}\\
\{J_{F_{\k,Lx}}(x): \k\in \RR_{>0}^m, x\in \VV_{>0}(f_\k)\}&=\{Q_F(w,h): (w,h)\in (\ker(C)\cap\RR^m_{>0})\times\RR_{>0}^n\}.
\end{align*}

We conclude this section by reminding the reader of a series of well-known definitions from the theory of reaction networks that will be used later on (see \cite{feinberg2019foundations} for a more detailed discussion of these properties):
\begin{itemize}
\item The \term{linkage classes} of a reaction network are its connected components as a digraph. Its \term{strong linkage classes} are the maximal strongly connected subdigraphs. 
Among the strong linkage classes, one  distinguishes the \term{terminal strong linkage classes}, as those for which there is no edge from a node inside of the class to a node outside of it, in the original network.

\item A reaction network is \term{weakly reversible} if
all connected components of the underlying digraph are strongly connected. 

\item  The \term{deficiency} of a  reaction network is the nonnegative integer
$\delta=c-\ell-s$, 
where $c$ is the number of complexes and $\ell$ is the number of linkage classes.

\item A reaction network is called \term{conservative} if   $\ker(N^{\top})\cap\RR^n_{>0} \neq \varnothing$, that is, the row span of $L$ contains a positive vector (equivalently the Euclidean closure of $\mathcal{P}_b$ is a compact set for all $b$ \cite{benisrael}).

\item  If $M\in \ZZ_{\geq 0}^{n\times m}$ (so that  zeros of $f,F$ can be considered in $\RR_{\geq 0}^n$), a reaction network  is said to \term{lack relevant boundary steady states} if there does not exist $(\k,b)\in\RR^m_{>0}\times\RR^d$ and $x\in \RR^n_{\geq 0} \setminus\RR^n_{>0}$ such that $F_{\k,b}(x)=0$ and 
 $\P_b\cap\RR^n_{>0}\neq\varnothing$.
\end{itemize}

\section{Nondegeneracy, nonemptiness, dimension and finiteness}\label{sec:theorems}

This section is devoted to the key theorems about generic nonemptiness and dimension of the positive steady state variety   and its intersection with the stoichiometric compatibility classes. 
The results are a consequence of the   general theorems on (augmented) vertically parametrized systems proven in \cite{FeliuHenrikssonPascual2023}, which are rewritten here for this more restricted context and adapted to the reaction network language. 
A version of the main result there can be found in \Cref{sec:thm_proofs}. The main take-home message is that the existence of a \emph{nondegenerate} zero is enough to ensure that the positive varieties behave ``nicely''. 

\subsection{Terminology: nondegeneracy and genericity}
\label{subsec:terminology}

We begin the section by establishing some terminology that will be important throughout the remainder of the paper.

In general, for a polynomial system $g=(g_1,\ldots ,g_\ell)$ in variables $x=(x_1,\ldots ,x_n)$, a zero $x^*$ of $g$ is said to be  \term{nondegenerate} if the Jacobian matrix $J_g(x)=\big(\frac{\partial g_i}{\partial x_j}\big)_{i,j}$ of $g$ has full rank when evaluated at $x^*$, and \term{degenerate} otherwise. 

We extend this notion to our parametrized steady state systems by saying that $f$ is \term{nondegenerate} if there exists $\k\in\RR^m_{>0}$ such that $f_{\kappa}$ has a nondegenerate positive zero $x^*\in\VV_{>0}(f_{\k})$, and \term{degenerate} otherwise. Similarly, we say that $F$ is \term{nondegenerate} if there exists $(\k,b)\in\RR^m_{>0}\times\RR^d$ such that $F_{\k,b}$ admits a nondegenerate positive zero $x^*\in\VV_{>0}(F_{\k,b})$.

As customary in reaction network theory, we adopt the terminology that a steady state $x^*\in \VV_{>0}(f_\k)$ is called \term{nondegenerate} if it is nondegenerate as a zero of $F_{\k,Lx^*}$, and \term{degenerate} otherwise. Similar to \cite{Banaji:Feliu}, we call $x^*$ \term{weakly nondegenerate} if it is a nondegenerate zero of $f_{\k}$ and \term{strongly degenerate} otherwise. 

Furthermore, we extend this terminology to networks by saying that a network is \term{(weakly) nondegenerate} if there exists $\k\in\RR^m_{>0}$ and $x^*\in\VV_{>0}(f_{\k})$ such that $x^*$ is a (weakly) nondegenerate steady state. Similarly, we say the network is  \term{(strongly) degenerate} if all steady states are (strongly) degenerate  for all choices of rate constants.

Finally, we say that a property holds for \term{generic} parameters in a set $\A\subseteq\RR^{\ell}$ when this property holds 
in a nonempty Zariski open subset of $\A$. 
When $\A$ is $\RR^m_{>0}$ or $\RR^m_{>0}\times \RR^d$,  this implies that the property holds in a set with nonempty Euclidean interior, and 
outside a subset of $\A$ of Lebesgue measure zero, and hence is robust against small perturbations of the parameter values.

\subsection{Nonemptiness and dimension}
In the notation of \Cref{sec:notations_defs},   
the steady state system $f$ has $s$ linearly independent entries  and $n$ variables, while the augmented steady state system has $n$ entries and variables. One could therefore expect that $\dim(\VV_{>0}(f_\k))=n-s$ and $\dim(\VV_{>0}(F_{\k,b}))=0$ as semialgebraic sets. This is certainly the case in typical examples, but it is not hard to construct examples where the expectation does not hold true. 
The following theorems give us precise tools to describe
the geometry of $\VV_{>0}(f_\k)$ and  $\VV_{>0}(f_\k)\cap \P_b$ for  generic choices  of parameters in the solvability loci  $\Z$ and $\Z_{\rm cc}$.

\begin{theorem}[Expected dimension of steady state varieties]
\label{thm:fnet}
For a reaction network with stoichiometric matrix $N \in \ZZ^{n\times m}$ and   kinetic matrix $M\in \ZZ^{n\times m}$,  consider the steady state system  $f$ as in \eqref{eq:f} for $C\in \RR^{s\times n}$ of full rank $s=\rk(N)$ such that $\ker(N)=\ker(C)$.  Let $\Z$ be the solvability locus of $f$ and
assume that $\ker(C) \cap \RR^m_{>0}\neq \varnothing$. The following statements are equivalent:
\begin{enumerate}[label=(\roman*)]
 \item  $\rk (C \diag(w)M^\top)=s$ for some $w\in \ker(C)$.
\item $f_\k$ has   a nondegenerate zero in $(\CC^*)^n$ for some $\k\in \CC^m$.
\item For generic $\k\in \Z$, 
 all zeros of $f_{\k}$ in $(\CC^*)^n$ are nondegenerate.
\item  $\Z$ has nonempty Euclidean interior in $\RR^m_{>0}$.
\item $\Z$   is not contained in a hypersurface of $\RR^m$. 
\item $\VV_{\CC^*}(f_\k)$ has pure dimension $n-s$ for at least one $\k\in \CC^m$. 
\end{enumerate}
Furthermore, the following holds:
\begin{itemize}
    \item If the equivalent statements hold, then  $\VV_{\CC^*}(f_\k)$ and 
    $\VV_{\RR^*}(f_\k)$ have pure dimension $n-s$ for generic $\k\in \Z$, and the same is true for $\VV_{>0}(f_\k)$ as a semialgebraic set. Additionally, all zeros of  $f_\k$ in $\RR^n_{>0}$ are nondegenerate for generic $\k\in \Z$. 

    \item If the equivalent statements do not hold, then $\VV_{\CC^*}(f_\k)=\varnothing$ for generic $\k\in \RR^m_{>0}$, and when not empty, $\dim \VV_{\CC^*}(f_\k)> n-s$ and all zeros of $f_{\k}$ are degenerate. 
    \item The ideal generated by $f_\k$ in $\CC[x^\pm]$ is radical for generic $\k\in \CC^m$. 
\end{itemize}
\end{theorem}

The proof of this result can be found in  \Cref{sec:thm_proofs}.

Note that the properties (i)-(vi) in \Cref{thm:fnet} together with the assumption $\ker(C) \cap \RR^m_{>0}\neq \varnothing$ precisely characterize the property of the steady state system $f$ being \emph{nondegenerate} and the network being \emph{weakly nondegenerate}, in the sense of \Cref{subsec:terminology}.

\begin{example}
With the matrices given in \Cref{ex:calcium_network2}, it holds that 
 \[\rk(C\diag(w)M^\top)=3\]
 for $w=(1,1,1,2,1,1)\in\ker(C)\cap\RR^m_{>0}$. 
Hence condition (i) in \Cref{thm:fnet} is satisfied, and we conclude that $f$ is nondegenerate and the network is weakly nondegenerate. In particular
the set $\Z\subseteq \RR^6_{>0}$ of the $\k$'s for which the network has positive steady states has nonempty Euclidean interior, and  for generic $\k$ in $\Z$, all 
positive zeros of $f_\k$ are nondegenerate and
$\VV_{>0}(f_\k), \VV_{\CC^*}(f_\k), \VV_{\RR^*}(f_\k)$ have dimension $1$.
\end{example}

\begin{example}
\label{ex:deg}
For the following reaction network with mass-action kinetics
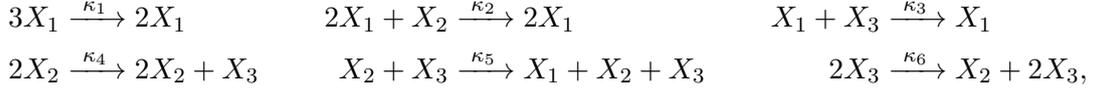
\begin{align*}
3X_1 & \ce{->[\k_1]} 2X_1 & 2X_1+X_2 & \ce{->[\k_2]} 2X_1  &  X_1+X_3 &  \ce{->[\k_3]} X_1 \\
 2X_2 &  \ce{->[\k_4]} 2X_2 + X_3  &X_2+X_3 & \ce{->[\k_5]} X_1+X_2+X_3 &  2X_3 & \ce{->[\k_6]} X_2+2X_3,
\end{align*}
we have $s=3$, and the defining matrices are
{\small \[C=N=\begin{bmatrix}
-1 & 0 & 0 & 0 & 1 & 0 \\ 0& -1 & 0 & 0 & 0 & 1 \\   0 & 0 & -1   & 1& 0 & 0
\end{bmatrix}  \qquad \text{and}\qquad M =  \begin{bmatrix}
3 & 2 & 1 & 0 & 0 & 0 \\
0 & 1 & 0 & 2 &1 & 0\\
0 & 0 &1  &0 & 1 & 2
\end{bmatrix}\, .  \]}
Any $w\in \ker(C)$ is of the form $w=(u_1,u_2,u_3,u_3,u_1,u_2)$ for some $u_1,u_2,u_3\in \CC$. But then
\[ \det(C\diag(w)M^\top) =\det \begin{bmatrix} -3 u_{1} & u_{1} & u_{1} 
\\
 -2 u_{2} & -u_{2} & 2 u_{2} 
\\
 -u_{3} & 2 u_{3} & -u_{3} 
\end{bmatrix} =0.
\] 
Hence, condition (i) in  \Cref{thm:fnet} does not hold. We conclude that $f$ is degenerate, the network is strongly degenerate, and none of the statements (i)-(vi) hold. In particular, any positive steady state is degenerate, 
and the positive steady state variety is generically empty.
\end{example}

\begin{theorem}[Finiteness of the number of stoichiometrically compatible steady states]
\label{thm:Fnet}
For a reaction network with stoichiometric matrix $N \in \ZZ^{n\times m}$ and kinetic matrix $M\in \ZZ^{n\times m}$, consider the augmented steady state system  $F$ as in \eqref{eq:F} for $C\in \RR^{s\times n}$ of full rank $s=\rk(N)$ such that $\ker(N)=\ker(C)$ and $L\in \RR^{d\times n}$  with $LN=0$ and $d=n-s$. Consider the solvability locus $\Z_{\rm cc}$ of $F$ and
assume $\ker(C) \cap \RR^m_{>0}\neq \varnothing$. 
The following statements are equivalent:
\begin{enumerate}[label=(\roman*)]
\item For some $w\in \ker(C)$ and $h\in (\CC^*)^n$ it holds that
\[\rk\left[\! \begin{array}{c}C \diag(w)M^\top \diag(h) \\ L \end{array}\!\right]=n. \]

\item $F_{\k,b}$ has   a nondegenerate zero in $(\CC^*)^n$ for some $(\k,b)\in \CC^m\times \CC^d$. 
\item For generic $(\k,b)\in \Z_{\rm cc}$, 
 all zeros of $F_{\k,b}$ in $(\CC^*)^n$ are nondegenerate. 
\item $\Z_{\rm cc}$ has nonempty Euclidean interior in $\RR^m_{>0}\times \RR^d$.
\item $\Z_{\rm cc}$   is not contained in a hypersurface of $\RR^m\times \RR^d$. 
\item $ \VV_{\CC^*}(F_{\k,b})$ is nonempty and finite for at least one $(\k,b)\in \CC^m\times \CC^d$.
\end{enumerate}

\smallskip
\noindent
Furthermore:
\begin{itemize}
    \item If the equivalent statements hold, then for generic $(\k,b)\in \Z_{\rm cc}$, it holds that
    \[0<\# (\VV_{>0}(f_\k)\cap \P_b) <\infty\] 
    and all points of $\VV_{>0}(f_\k)\cap \P_{b}$ are nondegenerate steady states.
    \item The ideal generated by $F_{\k,b}$ in $\CC[x^\pm]$ is radical for generic $(\k,b)\in \CC^m\times \CC^d$. 
\end{itemize}
\end{theorem}

The proof is analogous to that of \Cref{thm:fnet}, and both can be found in \Cref{sec:thm_proofs}, where additional properties of $f$ and $F$ are given. We note that $\ker(C) \cap \RR^m_{>0}\neq \varnothing$ together with any of the equivalent conditions (i)-(vi) in \Cref{thm:Fnet} characterize when both the network and the augmented steady state system $F$ are \emph{nondegenerate} in the sense of \Cref{subsec:terminology}.

\Cref{thm:Fnet} tells us that $\VV_{>0}(f_\k)\cap \P_b$ is either generically empty, or finite for generic parameter values in $\Z_{\rm cc}$. We obtain the following consequence. 

\begin{corollary}
\label{cor:infinite}
The set  $\left\{ (\k,b)\in\RR^m_{>0}\times L(\RR^n_{>0}) : \# (\VV_{>0}(f_\k)\cap \P_b)=\infty \right\}	$
is contained in a proper algebraic variety, and hence always 
has empty Euclidean interior in $\RR^m_{>0}\times L(\RR^n_{>0})$. 
\end{corollary}

It is straightforward to see that condition (i) in \Cref{thm:Fnet} implies condition (i) in \Cref{thm:fnet}.  This gives rise to the following corollary, which explains the choice of terminology of ``weak'' nondegeneracy and ``strong'' degeneracy.  Note that the converse of the corollary is not necessarily true, as shown by \Cref{ex:nondegenerate_but_degenerate_wrt_L}.

\begin{corollary}
\label{cor:fimpliesf}
If the augmented steady state system is nondegenerate, then so is the steady state system.  If a network is nondegenerate, then it is also weakly nondegenerate. If a network is strongly degenerate, then it is also degenerate.
\end{corollary}

\begin{example}
\label{ex:nondegenerate_but_degenerate_wrt_L}
Consider the   reaction network
\[ X_1 + X_2 \ce{->[\k_1]} X_1 \qquad X_2 \ce{->[\k_2]} 2X_2\, , \]
the steady states of which are described by the single parametric polynomial 
	 \[f =- \k_1 x_1 x_2 + \k_2x_2 = x_2(-\k_1x_1+ \k_2)\, .\]
We obtain that $\VV_{>0}(f_\k)=\{(x_1,x_2)\in \RR^2_{>0} : x_1=\tfrac{\k_2}{\k_1}\}$  for all $\k\in \RR^2_{>0}$.  In particular, $\Z=\RR^2_{>0}$, so $f$ is nondegenerate and the network is weakly nondegenerate by \Cref{thm:fnet}. 
As the stoichiometric compatibility classes are defined by the equation $x_1=b$ for $b>0$, $\VV_{>0}(f_\k)\cap \P_{b}\neq \varnothing$ only if $\tfrac{\k_2}{\k_1}=b$, and hence $\Z_{\rm cc}$ has empty Euclidean interior.   We conclude that both the system $F$ and the network are degenerate by \Cref{thm:Fnet}. 
\end{example}

\Cref{thm:fnet} and \Cref{thm:Fnet} concern  the \emph{generic} behavior of the zero sets, but they do not preclude pathological behaviors from arising for specific choices of parameters. 
Examples of what these behaviors can be are given next.

\begin{example}
\label{ex:somedegpars1}
The reaction network
    \[ 2X_1 + X_2 \ce{->[\k_1]} 3X_1 \qquad X_1+2X_2 \ce{->[\k_2]} 2X_1+X_2 \qquad X_1+X_2 \ce{->[\k_3]} 2X_2 \]
gives rise to the system $f=x_1x_2(\k_1x_1+\k_2x_2-\k_3)$ and the sets $\P_b$ are defined by the equation $x_1+x_2=b$.   As $\VV_{>0}(f_\k)\cap \P_{b}$ is defined by the intersection of two lines, it is straightforward to verify that
$\Z_{\rm cc}$ has nonempty Euclidean interior. Hence, 
the network is nondegenerate and $F_{\k,b}$ has finite positive zero sets for generic $(\k,b)\in\Z_{\rm cc}$.  
Nevertheless, for $\k=(1,1,1)$, $\VV_{>0}(f_\k)\cap \P_{b}=\varnothing$ if $b\neq 1$, whereas $\# (\VV_{>0}(f_\k)\cap \P_{b})=\infty $ for $b=1$.   This only happens nongenerically,  
namely for parameters of the form
$((\k_1,\k_1,\k_1b),b)$ with $\k_1,b>0$. 
\end{example}

\begin{example}
\label{ex:somedegpars2}
For the reaction network
\[
    3X_1+X_2\ce{->[\kappa_1]}  4X_1 \qquad
2X_1+X_2 \ce{->[\kappa_2]}  3X_2  \qquad
X_1+X_2 \ce{->[\kappa_3]}  2X_1,\]
we have $n=2$, $s=1$, and the steady state  system is
$f=x_1x_2(\kappa_1x_1^2-2\kappa_2x_1+\kappa_3)$, which is generically $1$-dimensional.
For parameters satisfying $\kappa_2^2=\kappa_1\kappa_3$, all positive zeros of $f_{\k}$ are degenerate but $\VV_{\CC^*}(f_\k)$,  $\VV_{\RR^*}(f_\k)$, and $\VV_{>0}(f_\k)$  nonetheless have pure dimension $1$. Hence, degeneracy of all positive steady states for a choice of parameters does not necessarily imply that the dimension of the steady state variety is higher than expected for that parameter value.  
\end{example}

\begin{example}
\label{ex:finite_deg}
For the reaction network
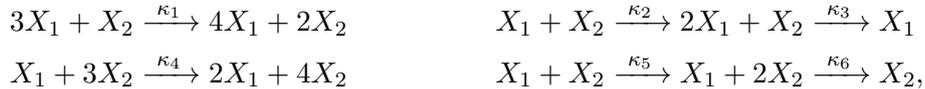
\begin{align*}
3X_1+X_2 & \ce{->[\kappa_1]}  4X_1+2X_2 & 
X_1+X_2 &  \ce{->[\k_2]} 2X_1+X_2 \ce{->[\k_3]} X_1 \\  
X_1+3X_2 & \ce{->[\kappa_4]}  2X_1+4X_2 &
X_1+X_2 &  \ce{->[\k_5]} X_1+2X_2 \ce{->[\k_6]} X_2,
\end{align*}
with stoichiometric and reactant matrices  
\[C=\begin{bmatrix} 1 & 1 & -1 & 1 & 0 & -1 \\
1 & 0 & -1 & 1 & 1 & -1 \end{bmatrix}\quad\text{and}\quad M=\begin{bmatrix}
3 & 1 & 2 & 1 & 1 & 1 \\ 
1 & 1 & 1 & 3 & 1 & 2 
\end{bmatrix}\, , \]
the steady state system $f$ is degenerate (condition (i) of \Cref{thm:fnet} fails). For $\kappa=(1,2,2,1,2,2)$ we have
\[f_\kappa = 
\left(\begin{array}{l}
x_1 x_2 (x_1^2+2-2x_1+x_2^2-2x_2) \\
x_1 x_2 (x_1^2-2x_1+x_2^2+2-2x_2) \end{array}\right)=\left(\begin{array}{l}x_1 x_2 ((x_1-1)^2+(x_2-1)^2)\\
x_1x_2((x_1-1)^2+(x_2-1)^2)\end{array}\right),\]
which has finite $\VV_{>0}(f_\kappa)$. Hence, statement (vi) in \Cref{thm:fnet} and \Cref{thm:Fnet} cannot be replaced by the existence of a choice of parameters for which the set of positive zeros is finite. 
\end{example}

\begin{remark}[Boundary steady states]
The results in this section concern the steady states with nonzero coordinates. 
By \cite[Theorem~3.19]{FeliuHenrikssonPascual2023}, the results on generic dimension extend from $(\CC^*)^n$ to $\CC^n$ if all the polynomials in $f_{\k}$ have a constant term involving a different parameter. This scenario arises when the network includes all \emph{inflow reactions}, that is, reactions of the form $0\rightarrow X_i$ for all species $X_i$ (in particular $s=n$). This is the case for \textit{Continuous-flow stirred-tank reactors (CFSTRs)}, see \cite{craciun-feinbergI}. For these networks, the existence of a nondegenerate zero of $f_\k$ ensures that the set of nonnegative steady states in $\RR^n_{\geq 0}$ 
has generically dimension $0$.
\end{remark}

\subsection{Computationally deciding on nondegeneracy}
\label{subsec:computational}
Checking condition (i) in \Cref{thm:fnet,thm:Fnet}  can be done computationally as follows. Let $G\in\RR^{m\times(m-s)}$ be a Gale dual matrix to $C$, in the sense that $\ker(C)=\im(G)$. Recall the matrices $Q_f(w)$ and $Q_F(w,h)$ from 
\eqref{eq:M_f_RN}.

\smallskip

\noindent
\textit{Consistency:} Nonemptiness of $\ker(C)\cap\RR^m_{>0}$ is equivalent to the feasibility of the system $\{Cx=0,x_1\geq 1, \dots,x_m\geq 1\}$, which can be checked with linear programming.

\smallskip
\noindent
\textit{Nondegeneracy of $f$:}
To check condition (i) of \Cref{thm:fnet}, we pick a random $u\in\RR^{m-s}$ (for some appropriate distribution), and compute $\rk(Q_f(Gu))$ with exact arithmetic. If the rank is $s$, then $f$ is nondegenerate. If  not, we view $u$ as indeterminate, compute the $s$-minors of $Q_f(Gu)$ in $\RR[u_1,\ldots,u_{m-s}]$, and use that $f$ is degenerate if and only if all minors are zero.

\smallskip
\noindent
\textit{Nondegeneracy of $F$:}
Similarly, condition (i) in \Cref{thm:Fnet} can be checked by first computing $\rk(Q_F(Gu,h))$ for a 
random $(u,h)\in\RR^{m-s}\times(\RR^*)^n$ with exact arithmetic. If the rank is $n$, then $F$ is nondegenerate; if not, we compute $\det(Q_F(Gu,h))$ as a polynomial in $\RR[u_1,\ldots,u_{m-s},h_1^\pm,\ldots,h_n^\pm]$, and use that $F$ is degenerate if and only if the determinant is zero. 

\smallskip

A \texttt{Julia} implementation of these computations, based on the computer algebra package \texttt{Oscar.jl} \cite{OSCAR} and the reaction network theory package \texttt{Catalyst.jl} \cite{Catalyst}, can be found in the Zenodo repository
\begin{center}
\url{https://doi.org/10.5281/zenodo.18063244}\,.  
\end{center}
\smallskip

As a demonstration of the applicability of our implementation, we consider the networks in the database ODEbase \cite{odebase}, modeling all of them with mass-action kinetics for simplicity. Out of 628 networks, we found that precisely 368 are consistent. Among these, 6 networks are strongly degenerate.
The other 362 networks are nondegenerate. The results of these computations are available in the aforementioned Zenodo repository.

The largest nondegenerate network in the database  is \texttt{BIOMD0000000014}, with $n=86$, $m=300$ and $d=n-s=9$, for which the nondegeneracy checks take less than 2~seconds\footnote{All computations were run on a MacBook Air with an Apple M2 chip and 16 GB of RAM.}.

\subsection{Network-theoretic conditions that ensure nondegeneracy}\label{subsec:classes}
In this section we go through a number of properties that play a central role in the reaction network theory literature, and prove that they imply
nondegeneracy of the network,  and hence  the equivalent properties in \Cref{thm:fnet} and \Cref{thm:Fnet}.

Weakly reversible networks are quite well understood; for example, they are known to admit positive steady states for all choices of $(\k,b)\in \RR^m_{>0}\times L(\RR^n_{>0})$  \cite{boros:ss}, and are conjectured to display strong dynamical behaviors such as persistence \cite{FEINBERG19872229} or bounded trajectories \cite{Anderson2011BoundednessOT}.
A somewhat unexpected property was revealed in \cite{boros2020weakly}, where  the authors provide several examples of  weakly reversible networks displaying infinitely many positive steady states in a  stoichiometric compatibility class. One such example is given next.

\begin{example}
\label{ex:BCY}
In \cite{boros2020weakly}, the authors considered the network 
 \[\begin{tikzcd}
	{\ce{2Y}} && {\ce{X + 3Y}} && {\ce{2X + Y}} \\
	{\ce{Y}} & {\ce{X + Y}} & {\ce{X + 2Y}} & {\ce{2X + 2Y}} & {\ce{2X}} & {\ce{3X}\rlap{,}}
	\arrow["{\kappa_1}"', from=1-1, to=2-1]
	\arrow["{\kappa_2}"', shift right=0.3ex, harpoon', from=2-1, to=2-2]
	\arrow["{\kappa_4}"', from=2-2, to=1-1]
	\arrow["{\kappa_3}"', shift right=0.3ex, harpoon', from=2-2, to=2-1]
	\arrow["{\kappa_7}"', from=2-3, to=2-4]
	\arrow["{\kappa_8}"', from=2-4, to=1-3]
	\arrow["{\kappa_6}"', shift right=0.3ex, harpoon', from=2-3, to=1-3]
	\arrow["{\kappa_5}"', shift right=0.3ex, harpoon', from=1-3, to=2-3]
	\arrow["{\kappa_{10}}"', from=2-5, to=2-6]
	\arrow["{\kappa_9}"', from=1-5, to=2-5]
	\arrow["{\kappa_{12}}"', shift right=0.3ex, harpoon', from=1-5, to=2-6]
	\arrow["{\kappa_{11}}"', shift right=0.3ex, harpoon', from=2-6, to=1-5]
\end{tikzcd}\]
with $s=n=2$, and fine tuned the  reaction rate constants in such a way that the two equations defining the steady states had a common factor and the other factors did not admit positive zeros. They obtained $\dim \VV_{\CC^*}(f_\k) =\dim \VV_{>0}(f_\k) =1$, and hence  infinitely many positive steady states in the only stoichiometric compatibility class $\RR^2_{>0}$. 
With this trick, they illustrated   that even for weakly reversible reaction networks, $\dim \VV_{>0}(f_\k)$ could be larger than expected for some choice of parameter values. 
\end{example}

Motivated by the example above, the authors of \cite{boros2020weakly}  posed the following question:
\smallskip
\begin{center}
\begin{minipage}[h]{0.9\textwidth}
\textit{Is it possible for a weakly reversible network  to have infinitely many positive steady states [in each positive stoichiometric compatibility class] for each choice of reaction rate constants in a [Euclidean] open set of the parameter space $\RR^m_{>0}$?}
\smallskip
\end{minipage}
\end{center}
We answered  the question in the negative in \Cref{cor:infinite} for \emph{any} network. 
So the condition on weak reversibility is superfluous in answering the question. However, for weakly reversible networks, we give the following strengthened answer. 

\begin{corollary}
\label{cor:wr_conservative}
Suppose that a reaction network satisfies at least one of the following:
 \begin{enumerate}[label=(\roman*)]
    \item It is weakly reversible and has mass-action kinetics.
    \item $M\in \ZZ_{\geq 0}^{n\times m}$, it is conservative, and lacks relevant boundary steady states.
\end{enumerate}
Then $\Z_{\rm cc}=\RR^m_{>0}\times L(\RR^n_{>0})$
and hence the network is nondegenerate.
\end{corollary}
\begin{proof}
The fact that $\Z_{\rm cc}=\RR^m_{>0}\times L(\RR^n_{>0})$ follows for case (i), from \cite{boros:ss}, and for case (ii), by a standard Brouwer fixed point argument, using that the Euclidean closure of $\P_b$ is compact and there are no steady states at the boundary (see, e.g., \cite[Lemma~1]{DuvallSontag2024}). 
\end{proof}

Two other important classes of networks for which our results have implications are those that are injective and those that satisfy the conditions of Feinberg's Deficiency One Theorem \cite[Theorem~4.2]{Feinberg1995existence}. 

A reaction network is called \textit{injective} if $f_{\k}$ is injective as a map $\mathcal{P}_b\to\RR^s$ for all $(\k,b)\in \RR^m_{>0}\times L(\RR^n_{>0})$ (see \cite{craciun-feinbergI,feliu2012preclusion,Muller2015injectivity}).
 The \textit{Deficiency One Theorem} requires that, when all linkage classes of a given network are considered separately, their deficiencies add up to exactly the deficiency of the network, none of them being higher than $1$. Additionally, each linkage class should have no more than one terminal strong linkage class.
These properties do not guarantee the existence of a positive steady state, so we need to require the network to be consistent. 

\begin{corollary}
Suppose that a reaction network satisfies $\ker(N)\cap \RR^m_{>0}\neq \varnothing$ and at least one of the following conditions:
 \begin{enumerate}[label=(\roman*)]
    \item It is injective.
    \item It fulfils the criteria of the Deficiency One Theorem.
\end{enumerate}
Then $\Z_{\rm cc}$ has nonempty Euclidean interior 
and hence the network is nondegenerate.
\end{corollary}
\begin{proof}
By \Cref{prop:Jacmatrices}(i), $\Z_{\rm cc}\neq \varnothing$. 
Any positive steady state is guaranteed to be nondegenerate in case (i) by \cite[Sec. 6]{semiopen}, see also \cite[Corollary~5.12]{feliu2012preclusion}, and in case (ii) by \cite[Theorem~4.3]{Feinberg1995existence}. The result now follows since  condition (i) in \Cref{thm:Fnet}  holds.
\end{proof}

\subsection{The kinetic and stoichiometric subspace}
We close the section with a discussion on how the dimension of the \emph{kinetic subspace} is related to nondegeneracy.

For a fixed choice of reaction rate constants $\k\in\RR^m_{>0}$, let 
$\Sigma_\k$ be the coefficient matrix of $N (\k \circ x^M)$ as a Laurent polynomial system in $\RR[x^\pm]^n$. 
The \emph{kinetic subspace} is 
\[S_\k := \im(\Sigma_\k), \]
which satisfies that the trajectories of \eqref{eq:ode} 
are contained in parallel translates of $S_\k$. In fact, this is the minimal vector subspace with this property. 
 It clearly holds that $S_\k \subseteq \im(N)$, but the inclusion might be strict. 
The stoichiometric subspace $S=\im(N)$ depends only on $N$, while the kinetic subspace also depends on the value of $\k$ and the kinetic matrix $M$.

In \cite{feinberg-invariant}, it is shown that if all linkage classes contain a unique terminal strong linkage class, then $S$ and $S_\k$ agree for all $\k$. Moreover, inequality of these spaces 
implies that all positive steady states are degenerate; this is immediate as if $S_{\k}\subsetneq S$, then the entries of $f_\k$ are linearly dependent  (see also \cite[Section 3.A.1]{feinberg2019foundations}). The converse  is not true, as it is possible to find degenerate systems $f$ with $S_{\k}=S$ for all $\k$ (e.g., \Cref{ex:deg}).
In \cite[Page 123]{feinberg2019foundations} it is stated that \textit{``when the kinetic subspace is smaller than the stoichiometric subspace, a positive equilibrium in a stoichiometric compatibility class will usually be accompanied by an infinite number of them''}. The following corollary   formalizes this observation    over the complex numbers, when $S_{\kappa}\subsetneq S$ holds for generic parameters.

\begin{corollary}
\label{cor:kinetic1}
Suppose that a reaction network satisfies $S_{\kappa}\subsetneq S$ for generic $\k\in\Z$. 
Then $f$ and $F$ are both degenerate. In particular, the network  is  strongly  degenerate, the  solvability loci  $\Z$ and $\Z_{\rm cc}$ have empty Euclidean interior, and 
    \[\#(V_{\CC^*}(F_{\k,b}))=\infty  \qquad \text{for all } \ (\k,b)\in \Z_{cc}. \]
\end{corollary}
\begin{proof}
As $S_{\kappa}\subsetneq S$ for generic $\k\in\Z$, 
all zeros of $f_\k$ are degenerate  for generic $\k\in\Z$. Thus condition (iii) in  \Cref{thm:fnet} does not hold, and $f$, and hence $F$,  are degenerate.   The rest of the conclusions follow from \Cref{thm:fnet,thm:Fnet}.  
\end{proof}

\begin{example}
Consider the following reaction network from \cite[Example 3.A.2]{feinberg2019foundations}
    \begin{align*}
X_2 \ce{<-[\k_1]} & X_1 \ce{->[\k_2]} X_3 &  X_2+X_3 &  \ce{->[\k_3]} 2X_1 
\end{align*}
with mass-action kinetics. This network  has defining matrices 
\[  N  = \begin{bmatrix} -1 & -1 & 2\\1 & 0 & -1 \\ 0 & 1 & -1\end{bmatrix}, \qquad \Sigma_{\k}= \begin{bmatrix} -\k_1-\k_2 & 2\k_3 \\ \k_1 & -\k_3 \\ \k_2 & -\k_3\end{bmatrix}\, \]
and $s=2$. The first row of $\Sigma_\k$ is linearly dependent to the two  bottom rows, and,  if the two bottom rows are linearly independent, that is, $\k_1\neq \k_2$, then $f_\k$ cannot have a zero in $(\CC^*)^2$. We conclude that $\Sigma_\k$ needs to have rank $1$ if $\k\in \Z$ and hence
$S_{\k}\subsetneq S$ for generic $\k\in \Z$. By \Cref{cor:kinetic1}, $f$ is degenerate, and  the network is strongly degenerate. 
\end{example}

Finally, the results from \cite{feinberg-invariant} regarding the stoichiometric and kinetic subspaces allow us to relate degeneracy and deficiency $\delta$. Let $\ell$ and $t$ denote the number of linkage classes and terminal strong linkage classes, respectively. 

\begin{corollary}
Suppose that a reaction network with mass-action kinetics satisfies $t-\ell>\delta$. Then both $f$ and $F$ are degenerate. In particular, the network   is strongly degenerate.
\end{corollary}
\begin{proof}
This follows from \Cref{cor:kinetic1}, as the assumption  implies $S_\k\subsetneq S$ for all $\k\in \RR^m_{>0}$ by statement (ii) in the theorem in   \cite[Section 6]{feinberg-invariant}. 
\end{proof}

\section{Generic absolute concentration robustness}\label{sec:ACR}
{\samepage
In this section, we use our new understanding of the steady state varieties to strengthen previous results on  absolute concentration robustness (ACR) and to clarify its generic behavior. 

A reaction network with a choice of power-law kinetics is said to have ACR for a certain species $X_i$ if it is consistent and, for each $\k\in \RR^m_{>0}$, the concentration $x_i$ of $X_i$  attains a unique value for all positive steady states  $x\in \VV_{>0}(f_\k)$ \cite{shinar-science}. 
The weaker notion of \emph{local} ACR was introduced in \cite{pascualescudero2020local}, and 
refers to $x_i$ attaining only a \emph{finite} number of values for all positive steady states   $x\in\VV_{>0}(f_\k)$ for each $\k\in\RR^m_{>0}$.
}

\begin{example}
\label{ex:IDH}
    Consider the following network with mass-action kinetics from \cite{shinar-science} describing the phosphorylation and dephosphorylation mechanism for the enzyme isocitrate dehydrogenase:
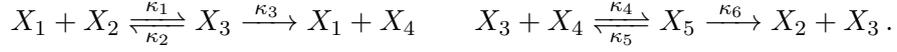
\[ 
		X_1 +X_2   \ce{<=>[\k_1][\k_2]} X_3 \ce{->[\k_3]} X_1 + X_4 \qquad 
             X_3 +X_4   \ce{<=>[\k_4][\k_5]} X_5 \ce{->[\k_6]} X_2 + X_3 \, .
\]
For  $\k\in\RR^6_{>0}$, the positive steady state variety is 
\[\VV_{>0}(f_\k)=\left\{ \left( x_1,x_2,\tfrac{\kappa_1}{\kappa_2+\kappa_3}x_1x_2,\tfrac{\kappa_3}{\kappa_4}\left( 1+\tfrac{\kappa_5}{\kappa_6}\right) ,\tfrac{\kappa _1\kappa _3}{\kappa_6(\kappa _2+\kappa_3)}x_1x_2\right) : x_1,x_2\in \mathbb{R}_{>0}\right\}  \subset \mathbb{R}^5_{>0},\]
from where it is clear that $x_4=\frac{\kappa_3}{\kappa_4}\big( 1+\frac{\kappa_5}{\kappa_6}\big)$ is constant for all $x\in \VV_{>0}(f_\k)$. The network has  ACR for $X_4$, but not for $X_1$, as the steady state value of $x_1$ is not constant. 
\end{example}

\begin{example}
\label{ex:localACR}  
For the network in  \Cref{ex:somedegpars2}, the value of $x_1$ at steady state is a positive root of the polynomial $\kappa_1x_1^2-2\kappa_2x_1+\kappa_3$, and hence attains one or two values whenever $\VV_{>0}(f_\k)\neq \varnothing$. Hence, the network displays local ACR for $X_1$. 
\end{example}

Many works have explored ACR, both from the experimental and the theoretical points of view, see for example \cite{ACR_Alon,ACR_Cap,ACR_AlgGeom,invariantsACR,ACRdim1,pascualescudero2020local, shinar-science}. However, no general procedure allows to detect it in practice, outside of specific contexts or under restrictive hypotheses.
Previous work in the literature has attempted to understand ACR at the level of the ideal generated by $f_\k$, often for fixed values of $\k$. One of the first sources of this is \cite{P-M}, and a more recent account of this approach can be found in  \cite{ACR_AlgGeom}. 
However, understanding the positive part of an algebraic variety for \emph{all} parameter values is 
a challenging problem as pathologies can easily arise for specific values of $\k$ (such as the variety having real dimension lower than expected). Several of these pathologies are discussed in \cite{ACR_AlgGeom}. 

As illustrated by \Cref{thm:fnet} and \Cref{thm:Fnet}, the behavior  of the positive steady state variety for \emph{generic} parameter values is nonpathological and resembles that of the complex variety. This leads us to introduce the concepts of 
 \emph{generic ACR} and \emph{generic local ACR},  and we show that these properties can be characterized in a more satisfactory way. 

In what follows, we let $\pi_i \colon \CC^n \rightarrow \CC$ denote the projection  onto  the $x_i$-coordinate.

\begin{defi}
Consider a  weakly nondegenerate reaction network with steady state system~$f$.
For a fixed $\k\in \Z$, we say that $\VV_{>0}(f_\k)$ has 
\begin{itemize}
\item  \term{local ACR} for $X_i$ if 
 $\# \pi_i(\VV_{>0}(f_\k))<\infty $, 
\item \term{ACR} for $X_i$ if 
 $\# \pi_i(\VV_{>0}(f_\k))=1$. 
\end{itemize}
We furthermore say that the network has 
\begin{itemize}
\item  \term{generic local ACR} for $X_i$ if 
 $\# \pi_i(\VV_{>0}(f_\k))<\infty $ for generic $\k\in \Z$,
\item \term{generic ACR} for $X_i$ if 
 $\# \pi_i(\VV_{>0}(f_\k))=1$ for generic $\k\in \Z$. 
\end{itemize}
If the properties hold for all $\k\in \Z$, then we say that the network has local ACR or ACR for $X_i$ respectively. 
\end{defi}

The following theorem establishes that
three criteria each completely characterize generic (local) ACR: 
a linear algebra  condition arising from \cite[Theorem~5.3]{pascualescudero2020local}, a property of the 
 elimination ideals over the coefficient field $\CC(\k)$ in the spirit of \cite[Proposition~3.8]{ACR_AlgGeom}, and the complex counterpart of the definition of (local) ACR.

{\samepage
\begin{theorem} 
\label{thm:localACR}\label{thm:localACR_poly}
Consider a reaction network with stoichiometric matrix $N\in \ZZ^{n\times m}$ and kinetic matrix $M\in \ZZ^{n\times m}$. Let $f=(f_1,\ldots,f_s)$ be the steady state system \eqref{eq:f} with  
$C\in \RR^{s\times m}$ of full rank $s$ and $\ker(N)=\ker(C)$.  Suppose that the network is weakly nondegenerate (i.e., that $f$ is nondegenerate). 
Given $i\in \{1,\dots,n\}$, the following are equivalent:
\begin{enumerate}[label=(\roman*)]
    \item   The network has generic local ACR for $X_i$.
     \item $\rk( (N\diag(w)M^\top)_{\setminus i}) < s$ for all $w \in \ker(N)$,  where {\footnotesize ${\setminus i}$} indicates removal of the $i$-th column.
     \item $\pi_i(\VV_{\CC^*}(f_\k))$ is a finite set for generic $\k\in \CC^m$.
\item For $I:=\langle f_1,\dots,f_s\rangle \subseteq \CC(\k)[x^\pm]$, the elimination ideal $I\cap \CC(\k)[x_i^{\pm}]$ is generated by one   nonconstant  polynomial. 
 \end{enumerate}
Furthermore:
\begin{itemize}
\item If the network does not have generic local ACR for $X_i$, then, for generic $\k\in \Z$,  $\VV_{>0}(f_\k)$ has no (local) ACR for $X_i$. 
\item If the network has generic local ACR for $X_i$
and for a specific $\k\in \Z$ every irreducible component of $\VV_{\CC^*}(f_\k)$ that intersects $\RR^n_{>0}$  contains a nondegenerate zero of $f_{\k}$, then $\VV_{>0}(f_\k)$ has local ACR for $X_i$. In particular, the network has local ACR for $X_i$ if $\rk(C\diag(v)M^\top)=s$ for \emph{all} $v\in\ker(C)\cap \RR^m_{>0}$. 
\end{itemize}
\end{theorem}
}

The proof of \Cref{thm:localACR} is given in \Cref{proof:localACR}, as it relies on a general result on augmented vertically parametrized systems. We emphasize first, that condition (ii) of \Cref{thm:localACR} can easily be rejected by taking one random choice of parameter values as in Subsection~\ref{subsec:computational}, and second, that conditions (iii) and (iv) of \Cref{thm:localACR}  refer to  the complex variety $\VV_{\CC^*}(f_\k)$; hence generic local ACR cannot occur if the equivalent property does not arise over $\CC^*$. 
The following corollary is a  direct consequence of \Cref{thm:localACR}.
 
\begin{corollary}
\label{cor:necessary_ACR}
Consider a  weakly nondegenerate reaction network with steady state system~$f$. 
If the set of $\k\in\RR^m_{>0}$ for which $\VV_{>0}(f_\k)$ has local ACR for $X_i$ has nonempty Euclidean interior, then the network has generic local ACR for $X_i$, and hence any of the equivalent statements (ii)-(iv) in \Cref{thm:localACR}  holds. 
    \end{corollary}

\begin{remark}
\label{rem:saturation}
Suppose $M\in \ZZ_{\geq 0}^{n\times m}$. In this case, we can replace $I=\langle f_1,\dots,f_s\rangle \subseteq \CC(\k)[x^\pm]$ in \Cref{thm:localACR_poly}(iv) with the saturation ideal 
\[\langle f_1,\dots,f_s\rangle: (x_1\cdots x_n)^\infty  \subseteq \CC(\k)[x]\, \]
as this ideal equals $I  \cap  \CC(\k)[x]$. 

In \cite[Proposition~3.8]{ACR_AlgGeom}, the ideal  $\langle f_{\k,1},\dots,f_{\k,s}\rangle \cap \CC[x_i] \subseteq \CC[x_i]$ is studied for fixed values of $\k\in\RR^m_{>0}$, as a sufficient condition for $\VV_{>0}(f_\k)$ to have ACR.  It is also given as a sufficient condition that  the elimination ideal of the saturated ideal is generated by a polynomial of the form $x_i-\alpha$ for some $\alpha$. We settle here that the study of the generator of the elimination ideal after saturating $\langle f_{\k,1},\dots,f_{\k,s}\rangle$  completely characterizes whether (local) ACR arises generically. 
\end{remark}

\begin{corollary}
\label{cor:ACRgen}
Consider a  weakly nondegenerate reaction network with steady state system~$f$. The following statements are equivalent:
\begin{enumerate}[label=(\roman*)]
\item
The network has generic  ACR for $X_i$.
\item For $I:=\langle f_1,\dots,f_s\rangle \subseteq \CC(\k)[x^\pm]$, the elimination ideal $I\cap \CC(\k)[x_i^{\pm}]$ is generated by one polynomial that has exactly one positive root for generic $\k\in \Z$. 
\end{enumerate}
\end{corollary}
\begin{proof}
Immediate from \Cref{thm:localACR_poly}. 
\end{proof}

Combining \Cref{thm:localACR}, \Cref{cor:ACRgen} and \Cref{rem:saturation}, we easily obtain \Cref{thmA:acr_results} from the introduction.

\begin{example} 
In \Cref{ex:localACR}, it holds that 
\[\big\{C\diag(w)M^\top: w\in \ker(C)\big\}=\left\{\begin{bmatrix} -2w_1+2w_2 &  0 \end{bmatrix}:  w_1,w_2\in \CC\right\}.\] 
We see that condition (i) in \Cref{thm:fnet} holds, so $f$ is nondegenerate. 
Removal of the first column gives a matrix of rank $0$. Hence, 
condition (ii) in \Cref{thm:localACR} is satisfied, and we conclude that the network  has generic local ACR for $X_1$. 
\end{example}

\begin{remark}
Clearly, having generic ACR is necessary for a weakly nondegenerate network to have ACR. 
\Cref{thm:localACR} gives also that lack of generic local ACR implies that generically, there is no ACR. 
We illustrate here that other relations between these concepts might not hold. 

We exploit an easy source of examples for ACR, which arise from networks with full rank $s=n$. When the number of positive steady states is finite,  uniqueness of steady states readily gives ACR. 
Although these are not the interesting cases in applications, they allow us to understand what phenomena are not to be expected. 
In particular:

\smallskip
\noindent
(1) \emph{Generic (local) ACR does not imply (local) ACR}:  The steady state system
\[f=(-\k_1 x_1 x_2 + \k_2 x_2^2  + \k_3x_2, \k_1 x_1 x_2 - \k_4 x_2^2 - \k_5 x_2), \]
which can be seen to arise from a network with mass-action kinetics, has  a solvability locus $\Z$ with nonempty Euclidean interior, and for generic $\k\in \Z$, there is precisely one positive zero, namely $\big(\tfrac{\k_2\k_5 - \k_3 \k_4}{\k_1(\k_2-\k_4)},\tfrac{\k_5-\k_3}{\k_2-\k_4}\big)$.  Thus, the network has generic ACR for $X_1$ and $X_2$. However, when $\k_2=\k_4$ and $\k_3=\k_5$, the zero set consists of the points 
$(\tfrac{\k_2x_2+\k_3}{\k_1},x_2)$ for $x_2>0$, and hence there is no ACR (nor local ACR).

\smallskip
\noindent
(2)  \emph{Absence of generic  ACR does not imply that, generically, there is no ACR}: 
The network with one species
\[ 3X_1 \ce{->[\k_1]} 2X_1 \qquad 2X_1 \ce{->[\k_2]} 3X_1 \qquad  X_1 \ce{->[\k_3]} 0  \qquad  0 \ce{->[\k_4]} X_1 \]
has steady state system $f= - \k_1x_1^3 + \k_2x_1^2 - \k_3x_1 + \k_4$, which is an arbitrary degree three polynomial with coefficients of fixed and alternating sign. 
Hence,  the solvability locus  $\Z$ has nonempty Euclidean interior. Furthermore, in a nonempty Euclidean open subset of $\Z$, $f$ has exactly one positive root, so there is ACR for $X_1$. However, also in a  nonempty Euclidean open subset of $\Z$, $f$ has three positive roots, and hence there is no ACR. As both presence and absence of ACR occur in  nonempty Euclidean open sets, none of the properties arise generically for this network. 
 
\smallskip
\noindent
(3)  \emph{Absence of generic local ACR does not preclude local ACR for a specific $\k$}:  The network with two species
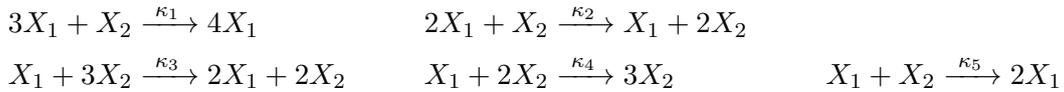
\begin{align*}
3X_1 + X_2  & \ce{->[\k_1]} 4X_1   & 2X_1+X_2  & \ce{->[\k_2]} X_1+2X_2    \\
X_1+3X_2 & \ce{->[\k_3]} 2X_1+2X_2 &    X_1+2X_2 & \ce{->[\k_4]} 3X_2 & X_1+X_2 \ce{->[\k_5]} 2X_1  
\end{align*}
 has $s=1$ and the steady state system is 
 \[f=x_1x_2(\k_1 x_1^2-\k_2 x_1 + \k_3 x_2^2 - \k_4 x_2  + \k_5) \, .
 \]
 For generic $\k$ in  the solvability locus  $\Z$, the zero set is a nonlinear curve,  and hence, there is no generic local ACR. However, for $\k=(1,2,1,2,2)$, we have $f_\k=(x_1-1)^2 + (x_2-1)^2$,  $\VV_{>0}(f_\k)$ consists of one (degenerate) point, and has trivially ACR for $X_1$ and $X_2$.  
 \end{remark}

The Zenodo repository of this paper (cf.~\Cref{subsec:computational}) also includes code for checking the rank condition in \Cref{thm:localACR}. When applying the condition to the networks in the database ODEbase taken with mass-action kinetics, we found 48 consistent nondegenerate networks that have generic local ACR but not full rank (which would trivially imply generic local ACR). For this check, we excluded species that do not appear in the reactant or product of any of the reactions. Only one of the $48$ networks was originally modeled with mass-action kinetics, which implies the remarkable fact that 
only one of the $72$ networks reported with mass-action kinetics has generic (nontrivial) local ACR. This network is studied in more detail next in \Cref{ex:auxin}.

\begin{example}
\label{ex:auxin}
The network \texttt{BIOMD0000000413} from \cite{auxin} offers a simple model of the degradation of the {\small (DII-)VENUS} reporter in response to the plant hormone auxin in a certain experimental setup, c.f. \cite[Fig. 2(B-C)]{auxin}. Specifically, the network is 
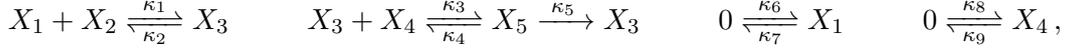
\begin{align*}
 X_1+ X_2  & \ce{<=>[\k_1][\k_2]} X_3 & 
X_3 + X_4 & \ce{<=>[\k_3][\k_4]} X_5  \ce{->[\k_5]} X_3 & 0    &  \ce{<=>[\k_6][\k_7]} X_1 & 
0  & \ce{<=>[\k_8][\k_9]} X_4\, ,
\end{align*}
where $X_1={\rm auxin}$, $X_2={\rm TIR1}$, $X_3={\rm auxinTIR1}$,  $X_4={\rm VENUS}$, $X_5={\rm auxinTIR1VENUS}$, and the kinetics is mass action. 
The network is nondegenerate and the condition (ii) in \Cref{thm:localACR} holds for $i=1$, which allows us to easily conclude that the network has generic local ACR for $X_1$. 
The polynomial $\k_7 x_1 - \k_6$ generates the elimination ideal $I\cap \CC(\k)[x_1^\pm]$ in \Cref{thm:localACR}(iv), and hence the network displays generic ACR for $X_1$. 

An easy check verifies that $\k_7 x_1 - \k_6$ is the sum of the first and second entries of $N (\k \circ x^M)$, and hence its specialization belongs to $I\cap \CC[x_1^\pm]$ for all $\k$. In particular, the network displays ACR for $X_1$. 

We remark that this network has deficiency $1$, but the species $X_1$ does not fulfill the sufficient conditions for ACR from \cite{shinar-science}. ACR can though be seen to arise  from the addition of  inflow and outflow reactions ($0 \ce{<=>} X_1$) to a network for which $X_1$ appears in a conservation law ($x_1+x_2- b$). 
\end{example}
 
\begin{example}
Network \texttt{BIOMD0000000167} from \cite{auxin}, modeled with mass-action kinetics for all reactions, including those originally modeled with Michaelis-Menten kinetics, satisfies the conditions for generic local ACR in $X_9$. In fact, we have generic ACR, since 
\[ \k_2 \k_4^2 \k_6^2 \k_7 \k_{10} \big( \k_{11}^2 + 2 \k_{11} \k_{13} + \k_{13}^2 \big)x_9^2\\
- \k_1 \k_3^2 \k_5^2 \k_8 \k_9 \big(\k_{12}^2 + 2 \k_{12} \k_{14} + \k_{14}^2\big)\in \langle f\rangle\cap \CC(\k)[x_9^{\pm}]\, , \] 
which clearly has a unique positive root for all $\k\in\RR^{14}_{>0}$.
\end{example}

\section{Nondegenerate multistationarity}\label{sec:multi}

A reaction network is said to admit \term{multistationarity} if there exists $(\k,b)\in\RR^m_{>0}\times\RR^d$ such that 
\[2 \leq \#\VV_{>0}(F_{\k,b})\, .\]
We say that the network admits \term{nondegenerate multistationarity} if in addition $\VV_{>0}(F_{\k,b})$ contains two  nondegenerate positive steady states.
The following theorem, which is strengthening of  \cite[Theorem~1]{Conradi2017identifying}, shows that, for a special type of networks, multistationarity and nondegenerate multistationarity go hand in hand. 
The concept of \term{dissipative} networks, which generalizes conservative networks, refers to 
networks for which, given $\k$ and $b$, there exists a compact set $K$ such that the trajectories of \eqref{eq:ode} in $\P_b$ eventually remain in $K$; see  \cite{Conradi2017identifying} for details. Recall the matrix $Q_F(w,h)$ given in \eqref{eq:M_f_RN}. 

\begin{theorem}
\label{thm:multidegree}
Consider a reaction network with kinetic matrix $M\in \ZZ_{\geq 0}^{n\times m}$. Assume that the network lacks relevant boundary steady states and is dissipative, and let  $F$ be the  augmented steady state system. 
Assume that  $\det(Q_F(w,h))$  attains both positive and negative signs  for $(w,h)\in  (\ker(N) \cap \RR^m_{>0})\times \RR^n_{>0}$. Then the network admits at least three nondegenerate positive steady states for some choice of parameters. 
\end{theorem}
\begin{proof}
By \cite[Theorem~1]{Conradi2017identifying}, there exists  $\varepsilon \in \{\pm 1\}$ (depending on $\rk(N)$ and a choice of a specific order of the rows of $J_{F_{\k,b}}(x)$) such that  the network admits multistationarity if 
\begin{equation}\label{eq:sign}
 \text{sign}(\det(J_{F_{\k^*,Lx^*}}(x^*))=\varepsilon 
\end{equation}
  for some $\k^*\in \RR^m_{>0}$  and $x^*\in \VV_{>0}(f_\k)$. More specifically, \cite[Theorem~1]{Conradi2017identifying} states that if \eqref{eq:sign} holds, then the network has at least  two positive  steady states for this  $\k^*$ and $b=L x^*$, and if  all positive steady states are nondegenerate, then there is an odd number of them. 

By \Cref{prop:Jacmatrices} and hypothesis, 
there exists 
$(w,h)\in (\ker(N) \cap \RR^m_{>0})\times \RR^n_{>0}$  such that the sign of $\det(Q_F(w,h))$ is $\varepsilon$. Then $(\k^*,b^*,x^*):=\phi(w,h)$ (with $\phi$ as in \eqref{eq:parametrization_of_incidence_variety_pos}) is such that $x^*$ is a nondegenerate positive steady state for the parameters $(\k^*,b^*)$, and there is multistationarity.  

On one hand, by the implicit function theorem, there exists a Euclidean ball $B\subseteq \RR^m_{>0}\times \RR^d$ containing $(\k^*,b^*)$, such that for all $(\k,b)\in B$, it holds that $F_{\k,b}$ has a nondegenerate zero $x$ with the sign of $\det(J_{F_{\k,b}}(x))$ being equal to $\varepsilon$.

On the other hand, \Cref{thm:Fnet} and the existence of a nondegenerate positive steady state tells us that all positive steady states are nondegenerate 
for all $(\k,b)$ in a nonempty Zariski open subset $\U\subseteq \Z_{\rm cc}$. 
We conclude that for any  $(\k,b)\in \U\cap B$ (where the latter set  is nonempty), all positive steady states are nondegenerate and  at least one   satisfies \eqref{eq:sign}. Hence there is an odd number of them and at least three, giving the statement. 
\end{proof}

Under the same hypotheses of \Cref{thm:multidegree}, it already follows from \cite[Theorem~1]{Conradi2017identifying} that the network admits at least two positive steady states for some choice of parameters, while nondegeneracy is not addressed. Our strengthening of this result is that we can guarantee the existence of three positive steady states, all of which will be nondegenerate.

Next, we give conditions that ensure that a network with the capacity for multistationarity also has the capacity for nondegenerate multistationarity.  This is related to the \emph{Nondegeneracy Conjecture} from \cite{Joshi2017small,shiuwolff:small,tang:one-dim}, which states that if $\VV_{>0}(F_{\k,b})$ is finite for all $(\k,b)\in\RR^m_{>0}\times\RR^d$, and has cardinality $\ell$ for some $(\k^*,b^*)\in\RR^m_{>0}\times\RR^d$, then there is some choice of parameters $(\k',b')\in\RR^m_{>0}\times\RR^d$ for which the network has $\ell$ nondegenerate positive steady states.
We analyze this problem in more generality from a topological point of view, with the help of the framework of vertically parameterized systems. We derive as a consequence the case $\ell=2$ of the conjecture for all nondegenerate networks.

\begin{theorem}
\label{thm:nondegeneracy_conjecture}
Suppose that  for  a given  $(\k^*,b^*)\in\RR^m_{>0}\times\RR^d$, a reaction network has at least
$\ell_\mathrm{ndeg}$ nondegenerate positive steady states, and at least $\ell_\mathrm{deg}>0$ degenerate  positive steady states, which are isolated points in $\VV_{>0}(F_{\k^*,b^*})$. Assume  furthermore that at least one of the following conditions holds: 
\begin{enumerate}[label=(\roman*)]
    \item The network is nondegenerate and $\VV_{>0}(F_{\k,b})$ is finite for all $(\k,b)$ in an open neighborhood of $(\k^*,b^*)$ in $\RR^m_{>0}\times\RR^d$.
    \item At least one of the degenerate positive steady states for $(\k^*,b^*)$  is an isolated point in the complex variety $\VV_{\CC^*}(F_{\k^*,b^*})$. 
\end{enumerate}
Then there exists a choice of parameters $(\k',b')\in\RR^m_{>0}\times\RR^d$ such that the network has 
at least 
\[\ell_\mathrm{ndeg} + {\rm min}(\ell_\mathrm{deg},2)\]
nondegenerate positive steady states.
\end{theorem}

\begin{proof} 
Consider the parametrization of the solvability locus $\Z_{\rm cc}$ of $F$,  obtained by composing $\phi$ in \eqref{eq:parametrization_of_incidence_variety_pos}  with the projection to parameter space:
\[
    \varphi \colon  (\ker(N) \cap \RR^m_{>0}) \times \RR^n_{>0}  \rightarrow  \RR^m\times \RR^d,\quad 
   (v,h)  \mapsto (v\circ h^M, Lh^{-1}),
\]
where $d=n-s$. An easy dimension count shows that both the domain and codomain of $\varphi$ are differential manifolds of dimension $m+d$. By construction, $\im(\varphi) = \Z_{\rm cc}$, and the fiber $\varphi^{-1}(\k^*,b^*)$ is homeomorphic to $\VV_{>0}(F_{\k^*,b^*})$.

Note that the network is nondegenerate also in case (ii):  the theorem of dimension of fibers \cite[Theorem~3.13, Corollary~3.15]{Mumford} says that if $\VV_{\CC^*}(F_{\k^*,b^*})$ has an isolated point, then $\VV_{\CC^*}(F_{\k,b})$ is nonempty and finite for generic $(\k,b)\in\CC^m\times\CC^d$, so the network is nondegenerate by \Cref{thm:Fnet}. Hence, $\Z_{\rm cc}$ is an $(m+d)$-dimensional semialgebraic set in both case (i) and (ii).

 \cite[Proposition~3.3]{FeliuHenrikssonPascual2023} establishes that when $(\k,b,x)=\phi(v,h)$, $x$ is  a nondegenerate positive steady state if and only if $(v,h)$ is a critical point of $\varphi$. In this case, $\varphi$ is a local diffeomeorphism around $(v,h)$, and hence 
$\varphi(v,h)$ is in the interior of $\varphi(U)$ for any 
  open neighborhood $U$ of $(v,h)$ in  $(\ker(N) \cap \RR^m_{>0}) \times \RR^n_{>0}$. 
We also have that the set  $\D_{\rm cc}$ of parameter values for which there is a degenerate positive steady state coincides with the set of critical values of $\varphi$, and the Zariski closure $H\subseteq \RR^m\times \RR^d$ of $\D_{\rm cc}$ is a proper algebraic variety  
under the assumption of nondegeneracy \cite[Proposition~3.4]{FeliuHenrikssonPascual2023}. So  
\[T:=\varphi^{-1}(H)\subseteq (\ker(N) \cap \RR^m_{>0}) \times \RR^n_{>0}\] 
is a proper Zariski closed subset.  
It follows that for any nonempty open set $U$ in $(\ker(N) \cap \RR^m_{>0}) \times \RR^n_{>0}$, the image $\varphi(U)$ has nonempty interior, as 
$U\setminus T\neq \varnothing$ consists  of regular points of $\varphi$, which are mapped to the interior of $\varphi(U\setminus T)$.  

Since $\ell_\mathrm{deg}>0$, we must have $(\k^*,b^*)\in H$. 
We want to show that there exists a choice of parameters $(\k',b')\in\RR^m_{>0}\times\RR^d$ such that 
\begin{equation}\label{eq:nondeg_goal}
\ell_\mathrm{ndeg}+ {\rm min}(\ell_\mathrm{deg},2)\leq \#\varphi^{-1}(\k',b') \qquad \text{and} \qquad (\k',b')\notin H. 
\end{equation} 

By assumption, the fiber $\varphi^{-1}(\k^*,b^*)$ contains $\ell_\mathrm{ndeg}$ isolated points $P_\mathrm{ndeg}:=\{\xi_1,\dots,\xi_{\ell_\mathrm{ndeg}}\}$ corresponding to nondegenerate positive steady states, and $\ell_\mathrm{deg}>0$ isolated points $P_\mathrm{deg} :=\{\delta_1,\dots,\delta_{\ell_\mathrm{deg}}\}$ corresponding to degenerate positive steady states, where, if in case (ii), we additionally  assume that $\delta_1$ corresponds to an isolated point in $\VV_{\CC^*}(F_{\k^*,b^*})$.

Let $P=P_\mathrm{ndeg}\cup P_\mathrm{deg}$ and 
$U_p$ for $p\in P$ be disjoint open balls in $(\ker(N) \cap \RR^m_{>0}) \times \RR^n_{>0}$   such that $U_p\cap\varphi^{-1}(\k^*,b^*)=\{p\}$.
Then $(\k^*,b^*)\in \bigcap_{p\in P} \varphi(U_p)$.  As noticed above, each $\varphi(U_p)$ has nonempty Euclidean interior
and for each $p\in P_\mathrm{ndeg}$, 
 $(\k^*,b^*)$ belongs to the interior of $\varphi(U_p)$. 
By shrinking each $U_{q}$ appropriately, we can without loss of generality assume 
\begin{equation}
\label{eq:shrinking}
    \varphi(U_{q})\subseteq \bigcap_{p\in P_\mathrm{ndeg}} \varphi(U_{p})\quad\text{for each $q\in P_\mathrm{deg}$}.
\end{equation}

If $V:=\bigcap_{p\in P} \varphi(U_p)$
has nonempty interior, then 
for all $(\k,b)\in V$, the fiber $\varphi^{-1}(\k,b)$ contains at  least $\ell_\mathrm{ndeg}+\ell_\mathrm{deg}$ distinct points: one in each $U_p$. Any $(\k',b')\in V\setminus H$ now satisfies  \eqref{eq:nondeg_goal}, and we are done. 
Note that this will be the case when $\ell_\mathrm{deg}=1$, in which case we have now shown that there are at least $\ell_\mathrm{ndeg}+1$ nondegenerate positive steady states. 

If the interior of $V$ is empty, then necessarily $\ell_\mathrm{deg}\geq 2$. 
If $(\k^*,b^*)$ is in the interior of 
 $\varphi(U_p)$ for at least one $p\in P_\mathrm{deg}$, say $\delta_{i}$, then the interior of $\varphi(U_{\delta_i}) \cap \varphi(U_{\delta_{j}}) \bigcap_{p\in P_\mathrm{ndeg}} \varphi(U_p)$
is nonempty for all $i\neq j$ and, as above, we are done by picking any such $j$. We therefore assume that $(\k^*,b^*)$ is at the boundary of  $\varphi(U_p)$ for all $p\in P_\mathrm{deg}$, and in particular for $\delta_1$. 
The rest of the proof will now focus on proving that, under this scenario, $\varphi$ is not injective on $U_{\delta_1}\setminus T$.

We first show that  there exists an open neighborhood  $U\subseteq U_{\delta_1}$ of $\delta_1$, such that the fibers of $\varphi_{|U}$ are finite. This is trivial in case (i). For case (ii), we use the upper semi-continuity (see \cite[02FZ]{stacks}) of the map $\xi \mapsto \dim_{\xi}( \varphi^{-1}_{\CC}(\varphi_{\CC}(\xi)))$, where $\dim_{\xi}$ denotes the local dimension at $\xi$, and $\varphi_\CC$ is the extension of $\varphi$ to $\CC$ as in \eqref{eq:parametrization_of_incidence_variety2} below. Since, by assumption, $\delta_1$ is isolated in $\varphi^{-1}_{\CC}(\k^*,b^*)$, 
this implies that there is an open neighborhood $U\subseteq U_{\delta_1}$ of $\delta_1$ such that $\varphi_{|U}^{-1}(\varphi_{|U}(\xi))$ is finite for all $\xi\in U$.

Aiming for a contradiction, we now assume that $\varphi$ is injective on $U\setminus T$.
As all fibers of $\varphi_{|U}$ are finite,  the Main Theorem in \cite{Blokh} tells us that $\varphi_{|U}$ is injective and hence an open map by the theorem of invariance of domain. Hence, $\delta_1$ is mapped to a point in the interior of $\varphi(U_{\delta_1})$, yielding a contradiction. 

We conclude that there exist at least two points in $U_{\delta_1}\setminus T$ mapping to some  
$(\k',b')\notin H$, each corresponding to a nondegenerate positive steady state. As $(\k',b')\in \bigcap_{p\in P_\mathrm{ndeg}}\varphi(U_p)$ by \eqref{eq:shrinking},  there are at least  $\ell_\mathrm{ndeg}+2$ nondegenerate positive steady states for $(k',b')$.
\end{proof} 

It remains an open question whether the lower bound of \Cref{thm:nondegeneracy_conjecture} can be improved to $\ell_\mathrm{ndeg}+\ell_\mathrm{deg}$. Nevertheless, we see that the existence of two isolated points, together with the conditions (i) or (ii), is sufficient for nondegenerate multistationarity. In particular, this proves the Nondegeneracy Conjecture for the $\ell=2$ case for all nondegenerate networks,   as stated in \Cref{thmA:nondegeneracy_conjecture} in the introduction.

\begin{corollary}
\label{cor:nondegeneracy_conjecture}
Suppose that for  a given choice of parameters $(\k^*,b^*)\in\RR^m_{>0}\times\RR^d$,  a reaction network has  at least two  positive steady states,
 which are isolated points in $\VV_{>0}(F_{\k^*,b^*})$,
and at least one of the following conditions holds: 
\begin{enumerate}[label=(\roman*)]
    \item The network is nondegenerate and $\VV_{>0}(F_{\k,b})$ is finite for all $(\k,b)$ in an open neighborhood of $(\k^*,b^*)$ in $\RR^m_{>0}\times\RR^d$.
    \item At least one of the  positive steady states for $(\k^*,b^*)$  is an isolated point in $\VV_{\CC^*}(F_{\k^*,b^*})$. 
\end{enumerate}
Then there also exists a choice of parameters  such that the network has at least two nondegenerate positive steady states.
\end{corollary}

\begin{remark}
For nondegenerate networks, condition (i)  of \Cref{thm:nondegeneracy_conjecture} and \Cref{cor:nondegeneracy_conjecture} is satisfied if $\VV_{>0}(F_{\k,b})$ is finite for all $(\k,b)\in\RR^m_{>0}\times\RR^d$ (as required in the original Nondegeneracy Conjecture of \cite{Joshi2017small}). Another sufficient condition for (i) is $\VV_{>0}(F_{\k,b})$ being finite \emph{and nonempty} for all $(\k,b)$ in an open neighborhood of $(\k^*,b^*)$ in $\RR^m_{>0}\times\RR^d$. 
We furthermore note that condition (ii) implies that the network is nondegenerate by \Cref{thm:summary_RN}. A sufficient condition for (ii) to hold is that $\VV_{\CC^*}(F_{\k^*,b^*})$ is finite.
\end{remark}

\begin{example}
\label{ex:counter_nondeg}
To understand some of the phenomena behind assumptions (i) and (ii) in \Cref{thm:nondegeneracy_conjecture}, we
consider  the network with $s=n=2$ and stoichiometric and kinetic matrix
\setcounter{MaxMatrixCols}{20}
\[N = \begin{bmatrix}
1 & -4 & 2 & -6 & 11 & -4 & 12 & -14 & 1 & -6 & 15 & -18 & 10 & 0\\
1 & -4 & 2 & -6 & 11 & -4 & 12 & -14 & 1 & -6 & 15 & -18 & 0 & 10
\end{bmatrix},
\]
\[ 
M = \begin{bmatrix}
5 & 4 & 3 & 3 & 3 & 2 & 2 & 2 & 1 & 1 & 1 & 1 & 1 & 1 \\
1 & 1 & 3 & 2 & 1 & 3 & 2 & 1 & 5 & 4 & 3 & 2 & 1 & 1
\end{bmatrix}.
\]

By subtracting the second row of $N$ to the first, we see that any positive steady state satisfies 
$(\k_{13}-\k_{14}) x_1x_2=0$. 
Hence, $\Z$ is generically empty, the network is (strongly) degenerate,  and all positive steady states for all choices of rate constants are degenerate. 
Taking $\kappa^*=(1,\ldots,1)$ gives 
\begin{equation}\label{f:counter_nondeg}
f_{\k^*} =
\left(\begin{array}{ll} 
x_1\, x_2 \left((x_1-1)^2+(x_2-1)^2\right)\left((x_1-1)^2+(x_2-2)^2\right)\\[0.2em]
x_1\, x_2 \left((x_1-1)^2+(x_2-1)^2\right)\left((x_1-1)^2+(x_2-2)^2\right)
\end{array}\right)
\end{equation}
for which $\VV_{>0}(f_{\k^*})=\{(1,1),(1,2)\}$, and hence has two (isolated) points. Therefore, we have multistationarity but not nondegenerate multistationarity.  
\end{example}

\begin{remark}
One might expect that the existence of two isolated positive steady states for some parameter values $(\k^*,b^*)$ for a nondegenerate network is sufficient for nondegenerate multistationarity.  
However, additional assumptions, such as (i) or (ii) in \Cref{thm:nondegeneracy_conjecture}, are needed to guarantee that multiple nondegenerate positive steady states arise for some perturbation of $(\k^*,b^*)$. 
To illustrate this, we consider a modification of    \Cref{ex:counter_nondeg} with full rank matrices 
\[N = \begin{bmatrix}
1 & -4 & 2 & -6 & 11 & -4 & 12 & -14 & 1 & -6 & 15 & -18 &0 & 10 & 0\\
1 & -4 & 2 & -6 & 11 & -4 & 12 & -14 & 1 & -6 & 15 & 0& -18 & 0 & 10
\end{bmatrix} \in \ZZ^{2\times 15},
\]
\[ 
M = \begin{bmatrix}
5 & 4 & 3 & 3 & 3 & 2 & 2 & 2 & 1 & 1 & 1 & 1& 1 & 1 & 1 \\
1 & 1 & 3 & 2 & 1 & 3 & 2 & 1 & 5 & 4 & 3 & 2 & 2& 1 & 1
\end{bmatrix} \in \ZZ_{\geq 0}^{2\times 15}.
\]

The polynomial $f_{\k,1}-f_{\k,2}$ yields 
$9(\k_{12} - \k_{13}) x_2 - 5(\k_{14} - \k_{15})=0$ at any zero in $(\CC^*)^2$. When $\k_{12}\neq \k_{13}$ and $\k_{14}\neq \k_{15}$, by isolating $x_2$, inserting the expression into $f_{\k,1}$ and clearing denominators and the factor $x_1\, x_2$, we obtain a degree $4$ polynomial   in $x_1$. 
Hence $\VV_{\CC^*}(f_\k)$ has generically four elements. As in addition $\ker(N)\cap \RR^2_{>0}\neq \varnothing$,  the steady state system $f$ is nondegenerate.  If either $\k_{12}= \k_{13}$ or $\k_{14}= \k_{15}$,  but not simultaneously, then there are no positive steady states. If 
$\k_{12}= \k_{13}$ and $\k_{14}= \k_{15}$, then $\VV_{\CC^*}(f_\k)$ has dimension $1$ and all zeros are degenerate, 
and for $\k^*=(1,\ldots,1)$, $f_{\k^*}$ coincides with \eqref{f:counter_nondeg}, and there are thus two degenerate and isolated positive steady states. 
Note that condition (ii) of \Cref{thm:nondegeneracy_conjecture} is  not satisfied, as all irreducible components of $\VV_{\CC^*}(f_{\k^*})$ are infinite. Condition (i) of \Cref{thm:nondegeneracy_conjecture} is not satisfied either, as $\VV_{>0}(f_{\k})$ is infinite for any $\kappa=(1,\ldots,1,a,a)$ with $0<a<1$.

We show next that   any perturbation of $\k^*$ yields a system without nondegenerate real zeros. To see this, we need to assume $\k_{12}\neq  \k_{13}$ and $\k_{14} \neq \k_{15}$, and consider the degree four polynomial  in $x_1$ introduced above. 
We reparametrize the polynomial using $\k_{12} = \epsilon + \k_{13}$ and $\k_{14} = a\cdot \epsilon + \k_{15}$ with $\epsilon\neq 0$ and $a>0$ (so $x_2$ is positive), such that it becomes:
\begin{multline*}
g  :=\epsilon^{4} \Big(6561 \k_{1} x_1^{4}-26244 \k_{2} x_1^{3}+(4050 \, a^{2} \k_{3}-21870 \, a\,  \k_{4}+72171 \k_{5}) x_1^{2} \\ \quad  +(-8100 \, a^{2} \k_{6}+43740 \, a\, \k_{7}-91854 \k_{8}) x_1+625\,  a^{4} \k_{9} \\  -6750 \, a^{3} \k_{10}+30375 \, a^{2} \k_{11}-65610 \, a\,  \k_{13}+65610 \k_{15} \Big). 
\end{multline*}
If we set $\k_i=1$ for all $i\not\in\{12,14\}$, $g$ has 
the roots
\[  1 \pm (\tfrac{5 \, a}{9} -2)  \,\mathrm{I}, \qquad 
1 \pm (\tfrac{5 \, a}{9}-1 )\mathrm{I}, 
\]
i.e., either $4$ complex roots, or $2$ complex roots and a double positive root, independently of $\epsilon$. Hence, no choice of $\epsilon\neq 0$ and $a> 0$ leads to nondegenerate positive steady states. If $a\not\in\{ \tfrac{9}{5},\tfrac{18}{5}\}$, then any small perturbation of the $\k_i$ for $i \notin \{12,14\}$ yields a polynomial with $4$ complex simple roots as well. 
If a perturbation for $a\in\{\tfrac{9}{5},\tfrac{18}{5}\}$ yielded positive real simple roots, then the same would be true after perturbing $a$, which we already shown is not the case. 

This reparameterization shows that any small perturbation of $\k^*$ will yield a system with no nondegenerate real zero. In particular, this illustrates that 
extra assumptions such as  (i) and (ii) in \Cref{thm:nondegeneracy_conjecture} are necessary if one aims at obtaining nondegenerate multistationarity by perturbing the given isolated zeros. 

We note that the system $f$ in this discussion, in fact, admits choices of parameters for which there are two nondegenerate positive steady states, but these are ``far'' from $\k^*$ and hence their existence is, in a sense, ``independent'' of the existence of two positive steady states for $\k^*$. 
\end{remark}

\section{Proofs of Theorem 3.1, Theorem 3.4 and Theorem 4.4}
\label{sec:thm_proofs}

In this final section we give the full theorem on augmented vertically parametrized  systems from  \cite{FeliuHenrikssonPascual2023} for completeness, and use it to derive \Cref{thm:fnet}, \Cref{thm:Fnet} and \Cref{thm:localACR}. 

An augmented vertically parametrized system is one of the form 
\[g = \left(C(\k \circ x^M),\, Lx-b\right)\in\CC[\k,b,x^\pm]^{s+\ell},  \qquad \rk(C)=s, \quad s\leq n, \quad 0\leq \ell \leq n-s. \]
With $\ell=0$, the steady state system is of this form (there is no linear part), and with $\ell=n-s$ so is the augmented steady state system. 
The complex incidence variety 
\[\mathcal{I}_g := \{(\k,b,x) \in \CC^m \times \CC^\ell \times (\CC^*)^n : g(\k,b,x)=0 \} \]
is nonsingular and 
admits a parametrization 
\begin{align}\label{eq:parametrization_of_incidence_variety2}
\phi \colon \ker(C)\times (\CC^*)^n \rightarrow \CC^{m}\times\CC^{\ell}\times (\CC^*)^n,\quad (w,h)\mapsto (w\circ h^M,Lh^{-1},h^{-1}),
\end{align}
analogous to \eqref{eq:parametrization_of_incidence_variety}. 
We let 
\[\Z_g = \{(\k,b) \in \RR^m_{>0}\times \RR^\ell : \VV_{>0}(g_{\k,b})\neq \varnothing \}. \]

The following theorem is \cite[Theorem~3.7]{FeliuHenrikssonPascual2023} (applied  to $\A=\RR^m_{>0}\times \RR^\ell$ and $\X = \RR^n_{>0}$ in  the notation of \cite{FeliuHenrikssonPascual2023}), combined with the considerations at the start of  Section~3.5 and Remark~2.4 in \cite{FeliuHenrikssonPascual2023}. { At the core of the proof is the fact that degenerate zeros of $g$ correspond to critical values of the projection of $\phi$ onto parameter space \cite[Proposition~3.3]{FeliuHenrikssonPascual2023}.  Additionally, we add the condition (setZC), which is equivalent to (setZ) by  \cite[Lemma~2.8]{FeliuHenrikssonPascual2023}, and (iso1), which is equivalent to (setZ) by the theorem of dimension of fibers (see, e.g., \cite[Theorem~2.3]{FeliuHenrikssonPascual2023}).}
 
\begin{theorem}
\label{thm:summary_RN}
For a real augmented vertically parametrized system 
\[ g=(C(\k \circ x^M),\, Lx-b)\in\RR[\k,b,x^\pm]^{s+\ell} \quad \text{with $s\leq n$ and $0\leq \ell \leq n-s$,}\]
assume that $\I_g \cap (\RR^m_{>0}\times \RR^\ell\times  \RR^n_{>0}) \neq \varnothing$. 
Consider the following statements:
\smallskip
\begin{itemize}[widest=(degXG),itemindent=*,leftmargin=*]
\item[{\rm (iso1)}] $g_{\k,b}$ has  an isolated zero in $(\CC^*)^n$ for some $(\k,b)\in \CC^m\times  \CC^\ell$.
\item[{\rm (deg1)}]  $g_{\k,b}$ has a nondegenerate zero  in $(\CC^*)^n$ for some $(\k,b)\in \CC^m\times  \CC^\ell$.
\item[{\rm (degX1)}] $g_{\k,b}$ has a nondegenerate zero  in $\RR^n_{>0}$ for some $(\k,b)\in \RR^m_{>0}\times \RR^\ell$.
\item[{\rm (degXG)}] There exists a nonempty Zariski open subset $\mathcal{U}\subseteq \I_g$ such that for all\\ $(\k,b,x)\in \mathcal{U}\cap (\RR^m_{>0}\times \RR^\ell\times  \RR^n_{>0})$, $x$ is a nondegenerate zero of $g_{\k,b}$. 
\item[{\rm (degAll)}] For generic $(\k,b)\in  \Z_g$,  all zeros of $g_{\k,b}$ in $(\CC^*)^n$ are nondegenerate.
\item[{\rm (setE)}] $\Z_g$ has nonempty Euclidean interior in $\RR^m_{>0}\times \RR^\ell$.
\item[{\rm (setZ)}] $\Z_g$ is Zariski dense in $ \CC^m\times  \CC^\ell$. 
\item[{\rm (setZC)}] $\{(\k,b) \in \CC^m \times \CC^\ell : \VV_{\CC^*}(g_{\k,b})\neq \varnothing \}$ is Zariski dense in $ \CC^m\times  \CC^\ell$. 
\item[{\rm (dim1)}] $\VV_{\CC^*}(g_{\k,b})$ has pure dimension $n-s-\ell$ for at least one $(\k,b)\in\Z_g$.
\item[{\rm (dimG)}] $\VV_{\CC^*}(g_{\k,b})$ has pure dimension $n-s-\ell$ for generic  $(\k,b)\in\Z_g$.
\item[{\rm (real)}]   For generic $(\k,b)\in\Z$, $\VV_{\RR^*}(g_{\k,b})$ has pure dimension $n-s-\ell$, and $\VV_{>0}(g_{\k,b})$ has dimension $n-s-\ell$ as a semialgebraic set. 
\item[{\rm (reg)}]  For generic $(\k,b) \in  \CC^m\times  \CC^\ell$, $\VV_{\CC^*}(g_{\k,b})$ is a nonsingular complex algebraic variety (in particular, different irreducible components do not intersect).
\item[{\rm (rad)}] $g_{\k,b}$ generates a radical ideal in $\CC[x^{\pm}]$ for generic $(\k,b) \in  \CC^m\times  \CC^\ell$, and $g$ generates a radical ideal in $\CC(\k,b)[x^{\pm}]$.
\end{itemize}
{\samepage
Then the following holds:
\begin{itemize}
\item The statements {\rm (deg1)}, {\rm (degX1)}, {\rm (degXG)}, {\rm (degAll)}, {\rm (setE)}, {\rm (setZ)}, {\rm (setZC)}, {\rm (dim1)}, {\rm (dimG)} are all equivalent to the condition 
\[\rk \begin{bmatrix}C \diag(w)M^\top \diag(h) \\ L\end{bmatrix} =s+\ell\quad\text{for some $(w,h)\in \ker(C) \times (\CC^*)^n$}.\]

\item Any of the statements mentioned above implies {\rm (real)} and {\rm (reg)}. 
\item  The statement {\rm (rad)} holds, independently of the other statements. 
\end{itemize}
}
\end{theorem}

\subsection{Proof of Theorem 3.1 and Theorem 3.4}
\label{proof:fnet}\label{proof:Fnet}
We consider the steady state system $f$, where $\ell=0$ and $\Z=\Z_f$ for \Cref{thm:fnet}, and the augmented steady state system $F$  with $\ell=n-s$ and $\Z_F=\Z_{\rm cc}$ for \Cref{thm:Fnet}. 
In both cases,  
the condition $\ker(C)\cap \RR^m_{>0}\neq \varnothing$ implies $\Z\neq \varnothing$ and $\Z_{\rm cc}\neq \varnothing$ by \Cref{prop:Jacmatrices}, and hence $\I_f\cap (\RR^m_{>0}\times  \RR^n_{>0})\neq \varnothing$ and $\I_F\cap (\RR^m_{>0}\times \RR^d \times  \RR^n_{>0})\neq \varnothing$. 
Therefore, we can apply  \Cref{thm:summary_RN}.

The equivalence between (i)--(vi) in \Cref{thm:fnet} and in \Cref{thm:Fnet} correspond to the first bullet point of \Cref{thm:summary_RN} by using (deg1), (degAll), (setE), (setZ), (dim1).

For \Cref{thm:fnet}, we now have that if the equivalent statements hold, \Cref{thm:summary_RN} gives   that (dimG) and (real) hold,
and so does the first bullet point. For the second bullet point, as (dim1) and (deg1) do not hold, we obtain the statement about dimension and degeneracy. As (setZC)  does not hold, the complex  varieties are nonempty only for parameters in a proper Zariski closed subset, hence  
generically empty.  

The last bullet point of  both \Cref{thm:fnet} and  \Cref{thm:Fnet} agree with that of \Cref{thm:summary_RN}. 
  
For \Cref{thm:Fnet}, the first bullet point  follows from  (real) and (degAll).\hfill\qed

\subsection{Proof of Theorem 4.4}
\label{proof:localACR}
To show the equivalence between (iii) and (iv), let $I_{\k}$ denote the specialization  of $I$ to $\k$. The closure theorem  gives us that  $\VV_{\CC^*}(I_{\k}\cap \CC[x_i^\pm])= \overline{\pi_i(\VV_{\CC^*}(f_\k))}$ for all $\k\in \CC^m$ \cite{cox2015ideals}. 
So $\pi_i(\VV_{\CC^*}(f_\k))$ is finite   if and only if $I_{\k}\cap \CC[x_i^\pm]$ is not the zero ideal. 
Given $G\subseteq  \CC(\k)[x^\pm]$ a Gr\"obner basis of $I$, it specializes to a Gr\"obner basis of $I_\k$ for generic $\k$. 
This gives that $I_{\k}\cap \CC[x_i^\pm]$ is generated by one polynomial for generic $\k$ if and only if $I\cap \CC(\k)[x_i^{\pm}]$ is, and from this the equivalence follows. 

With this in place, it is enough to show that (ii) $\Rightarrow$ (iii) $\Rightarrow$ (i) $\Rightarrow$ (ii).  As $f$ is nondegenerate, $\Z$ has nonempty Euclidean interior and is Zariski dense in $\CC^m$. 
Let $H$ be the augmented vertical system constructed by appending $x_i-c$ to $f$, and let $\Z_H\subseteq \RR_{>0}^m\times \RR$ be the subset of parameters $(\k,c)$ such that $H_{\k,c}$ has a zero in $\RR^n_{>0}$.  For each $\k\in \Z$, let $c_\k$ be the $x_i$-value of some $x\in \VV_{>0}(f_\k)$, which is nonempty by hypothesis. Then $(\k,c_\k)\in \Z_H$ as $H$ has a positive zero. Hence we have 
a dominant map  of varieties
\[\rho \colon \overline{\Z_H} \rightarrow \overline{\Z}=\CC^m, \qquad (\k,c) \mapsto \k, \]
where the overline refers to the Zariski closure in complex spaces. By the theorem of the dimension of fibers  \cite[Theorem~3.13]{Mumford}, we conclude that the fibers of $\rho$ have generically  dimension 
$\dim( \overline{\Z_H}) - m\,.$ 
By construction, for $\k\in \CC^m$, $\pi_i(\VV_{\CC^*}(f_\k))$ can be identified as a subset of $\rho^{-1}(\k)$, if nonempty. 

Condition (ii) is exactly the failure of the rank condition in \Cref{thm:summary_RN}, namely that the matrix 
\[ \begin{bmatrix} C \diag(w)M^\top \diag(h) \\  e_i \end{bmatrix} \in \CC^{(s+1)\times n}, \]
with $e_i$ is the canonical row vector with $1$ in the $i$-th entry and zero otherwise, has rank at most $s$ for all $w\in \ker(C)\subseteq \CC^m$ and $h\in (\CC^*)^n$. 
Hence, (ii) holds if and only if (setZC) does not hold, that is, $ \overline{\Z_H}$ is a proper Zariski closed subset of $\CC^m\times \CC$, that is, has dimension at most $m$. 
This in turn holds if and only if 
the fibers of $\rho$ have generically dimension $0$, in which case this implies that $\pi_i(\VV_{\CC^*}(f_\k))$ has generically dimension $0$, giving  (iii). 
 
As $\Z$ is Zariski dense, we readily have that  (iii) implies (i). Given $\k\in \RR^m_{>0}$, let $V_\k^{\rm nd}$ be the 
union of the irreducible components of $\VV_{\CC^*}(f_\k)$ that contain a nondegenerate positive zero of $f_\k$. Then 
\cite[Theorem~5.3]{pascualescudero2020local} states that (ii) is equivalent 
to $V_\k^{\rm nd}\cap \RR^n_{>0}$ having local ACR for $X_i$ for all $\k\in \RR^m_{>0}$ for which $V_\k^{\rm nd}\neq \varnothing$. 
By \Cref{thm:fnet}, all zeros of $f_\k$ are nondegenerate for generic $\k\in \Z$ and hence, generically, 
$\VV_{>0}(f_\k)=V_\k^{\rm nd}\cap \RR^n_{>0}$. 
Hence  (i) implies (ii).

We now show the bullet points of the statement. 
For the first bullet point,  if (i) does not hold, neither does (iii), and then  for generic $\k\in (\CC^*)^m$, $\pi_i(\VV_{\CC^*}(f_\k))$ is the complex line $\CC^*$ except for at most a finite number of points (by the Closure Theorem). As $\Z$ is Zariski dense, it follows that  $\pi_i(\VV_{>0}(f_\k))$ also is infinite for generic $\k\in \Z$ and hence, generically, there is no local ACR.  

For the first part of the second bullet point, the extra assumption implies $\VV_{>0}(f_\k)=V_\k^{\rm nd}\cap \RR^n_{>0}$ for the given $\k$. 
Then, the specialized version \cite[Theorem~4.11]{pascualescudero2020local} of 
\cite[Theorem~5.3]{pascualescudero2020local} gives that $\VV_{>0}(f_\k)$ has local ACR for $X_i$ if and only if 
the matrix 
$J_{f_\k}(x)_{\setminus i}$ has rank at most $s-1$ for all $x\in \VV_{>0}(f_\k)$. 
But this is guaranteed by (ii) and \Cref{prop:Jacmatrices}. 
The second part follows again from \cite[Theorem~5.3]{pascualescudero2020local} and \Cref{prop:Jacmatrices}, as the condition ensures $\VV_{>0}(f_\k)=V_\k^{\rm nd}\cap \RR^n_{>0}$ for all $\k\in \Z$. \hfill\qed
 
\section*{Acknowledgements}
EF and OH have been funded by the Novo Nordisk Foundation project with grant reference number NNF20OC0065582. 
BP has been funded by the European Union's Horizon 2020 research and innovation programme under the Marie Sklodowska-Curie IF grant agreement No 794627 and the Spanish Ministry of Economy project with reference number PID2022-138916NB-I00.
This work has also been funded by the European Union under the Grant Agreement no. 101044561, POSALG. 
Views and opinions expressed are those of the authors only and do not necessarily reflect those of the European Union or European Research Council (ERC). Neither the European Union nor ERC can be held responsible for them.
The authors thank Anne Shiu  and two anonymous referees for helpful feedback on earlier versions of the manuscript.

\bibliographystyle{siamplain}

\end{document}